\documentclass[12pt]{article}

\usepackage{amsmath}
\numberwithin{equation}{section}
\usepackage{amsfonts}
\usepackage{amsthm}
\usepackage{graphicx}
\usepackage[utf8]{inputenc}
\inputencoding{latin1}
\usepackage[usenames,dvipsnames]{color}
\usepackage[table]{xcolor}
\usepackage{multirow}
\usepackage{array}
\usepackage{float}
\floatstyle{ruled}
\newfloat{algorithm}{tbp}{loa}
\floatname{algorithm}{Algorithm}
\usepackage{caption}
\usepackage{subcaption}
\usepackage{epstopdf}
\usepackage[english]{babel}

\renewcommand\epsilon\varepsilon

\newtheorem{remark}{Remark}[section]
\newtheorem{theorem}[remark] {Theorem}

\renewcommand\theta\vartheta

\newcommand{\R}{\mathbb R}

\begin{document}

 \title{Modeling Networks with\\ a Growing Feature-Structure}
 
\author{Irene Crimaldi, Michela Del Vicario, 
Greg Morrison,\\ Walter Quattrociocchi, 
Massimo Riccaboni
\footnote{Alphabetic order. IMT Institute for Advanced Studies Lucca,
  Piazza San Ponziano 6, I-55100 Lucca, Italy.  E-mail:
  irene.crimaldi@imtlucca.it, michela.delvicario@imtlucca.it
  (corresponding author), greg.morrison@imtlucca.it,
  walter.quattrociocchi@imtlucca.it and massimo.riccaboni@imtlucca.it.
  Irene Crimaldi is a member of the Italian Group ``Gruppo Nazionale
  per l'Analisi Matematica, la Probabilit\`a e le loro Applicazioni
  (GNAMPA)'' of the Italian Institute ``Istituto Nazionale di Alta
  Matematica (INdAM)''.}  }

\maketitle

\begin{abstract}
\noindent We present a new network model accounting for {\em
  multidimensional assor\-ta\-ti\-vi\-ty}. Each node is characterized
by a number of features and the probability of a link between two
nodes depends on common features.  We do not fix a priori the total
number of possible features. The bipartite network of the nodes and
the features {\em evolves} according to a stochastic dynamics that
depends on three parameters that respectively regulate the
preferential attachment in the transmission of the features to the
nodes, the number of new features per node, and the power-law behavior
of the total number of observed features.  Our model also takes into
account a mechanism of {\em triadic closure}. We provide
theo\-re\-ti\-cal results and statistical estimators for the
parameters of the model. We validate our approach by means of
simulations and an empirical analysis of a network of scientific
collaborations.  \\

\noindent{\em keyword:} complex network, bipartite network,
assortativity, homophily, pre\-fe\-ren\-tial attachment, triadic closure.
\end{abstract}

\section{Introduction}
\label{intro}

Many complex systems are often described by means of a network of
interacting components, i.e. a set of nodes connected by links
\cite{barrat2008dynamical, caldarelli2007scale, ek, general-sna-book,
  wasserman}. A large number of scientific fields involve the study of
networks in some form: networks have been used to analyze
interpersonal social relationships, communication systems,
international trade, financial systems, co-authorships and citations,
protein interaction patterns, and much more. Therefore, formal
stochastic models and statistical techniques for the analysis of
network data have emerged as a major topic of interest in diverse
areas of study.  The distribution of the number of node's connections
is well approximated by a power-law in many contexts and preferential
attachment is generally accepted as the simplest mechanism that can
reproduce such a distribution \cite{barabasi2002, BA}. This basic
mechanism, however, is only one of the many forces that can contribute
to shape the evolution of complex networks. For instance, a social
network having power-law degree distribution is an exception rather
than the rule. In particular, preferential attachment is not able to
reproduce the formation of social groups, or communities, and the
composition of social circles. {\em Assortativity} (or assortative
mixing), called homophily in social networks, is defined as the
prevalence of network-links between nodes that are similar to each other
in some respect. Network theorists often analyze assortativity in
terms of a node's degree \cite{b2007, newm2003, pa2001, pi2008}.
Moreover, a large body of research in sociology and, more recently, in
economics, confirms the presence of a multidimensional assortativity
in socio-economic networks: homophily, along the lines of race and
ethnicity, age and sex, edu\-ca\-tion, professional background and
occupation, shapes complex networks such as friendship, marriage,
teamwork, co-membership, exchange and communication networks
\cite{bessi2014social, blau1984, block2014ns, brown2007word,
  currarini2009, ek, feld1982, goo2009, jackson2014, Kandel1978,
  kossinets2009ajs, koss-watts, Lazarsfeld1954, louch2000,
  Marsden1987, McPherson2001, Quattrociocchi2014, Verbrugge1977}. The
assortativity property has been also studied in citation networks: for
instance, in \cite{bramoulle2012jet} authors analyze the citations
among papers (the nodes of the network) published in journals of the
American Physical Society with respect to their PACS classification
codes, that represent the different research sub-fields. In formal
models assortativity is typically represented by partitioning nodes
into different classes (also called groups, clusters, or types)
related to some (observable or unobservable) features \cite{airoldi,
  bramoulle2012jet, gold, handcock, h2010, hoff, hs2003, hunter, kola,
  kri, nov, sn}. The assumption that each node can belong only to a
single class and/or the fact that the number of classes is finite and
fixed a priori as well as the number of the possible features restrict
their applicability.\\

\indent We contribute to this growing body of literature by
introducing a new sto\-cha\-stic model accounting for multidimensional
assortativity. The study of networks of papers, such as co-authorship
or citation networks \cite{barabasi2002evolution, bramoulle2012jet,
  golub, newman2004coauthorship}, is a particularly suitable
application of our model as the generative processes of features and
links are consistent with the basic aspects of the model: first, it is
a growing network process where nodes appear in chronological order
and do not exit; second, the links are established at the entrance of
the nodes and are unchangeable along time; third, each node exhibits
some features (for example, key-words, main topics, etc.) that are
unchangeable during time; finally, the set of the features grows in
time and the evolution of the nodes-features structure is interesting
exactly as the process of the link-creation among the nodes. Indeed,
the description of both phenomena is very important for the
understanding of the diffusion process of ideas and discoveries inside
a certain research field and among different research fields. Anyway,
as we will discuss at the end of this paper, our model can be easily
modified and/or enriched in order to get variants that better fit
networks of a different type.\\

\indent In particular, besides the link-creation mechanism, our model
provides a sto\-cha\-stic dynamics for the evolution of the
features. Differently from the above quoted works (see, for instance,
the model in \cite{kri} and the related discussion about the selection
problem for the dimension of the feature-space), we do not fix a
priori the total number of possible features but we allow the number
of observed features to grow in time. The bipartite nodes-features
network (i.e. the surrounding context) grows according to a
stochastic model that depends on three parameters that respectively
regulate the preferential attachment in the transmission of the
features to the nodes, the number of new features per node, and the
power-law behavior of the total number of observed
features. Concerning this point, the present paper may be considered
as a companion article to \cite{bol-cri-mon}. Indeed, both of them
provide an evolving dynamics for the feature-structure, but they also
show some differences. The main issue is that here we introduce a
parameter that tunes the preferential attachment in the transmission
of the features to the nodes; while in \cite{bol-cri-mon} authors only
consider a preferential attachment rule. Moreover, in that paper a
random ``fitness'' parameter which determines the node's ability to
transmit its own features to other nodes (see also \cite{bcpr-ibm}) is
attached to each node; while here we do not take into account fitness
parameters for nodes. \\

\indent Coming from a structural approach, differently from other
models which concentrate only on assortativity \cite{currarini2013wp,
  MGJ2009, palla-et-al, SCJ}, our model also accounts for the
principle, known as {\em triadic closure} or transitivity, according
to which, if A is a neighbor of B and B is a neighbor of C, then A and
C have a high chance to be neighbors. This principle is widely
supported on the em\-pi\-ri\-cal ground and it is at the basis of many
generative network models \cite{bianconi2014pre, ek, goo2009,
  ispolatov, jacksonrogers, koss-watts, louch2000, marsili,
  newman2003, palla2007quantifying, Rapoport53, sole, toivonen}. It is
worthwhile to note that the expression ``triadic closure''
conceptually refers to a link-formation process not depending on the
features of the nodes that get attached. However, also assortativity
can naturally induce closed triplets in the network and hence
evaluating assortativity and triadic closure separately sometimes may
be not easy. (For a further discussion on this issue, we refer to the
next Section \ref{estimation}.) Anyway models based on both mechanisms
produce more realistic networks.\\

\indent The paper is structured as follows. In Section \ref{ass} we
describe the basic assumptions of our model and the notation used
throughout the paper. In Section \ref{model} we present our stochastic
model, that involves a dynamics for the bipartite network of
nodes-features and the mechanism underlying the formation of the
unipartite (i.e. node-node) network. In Section \ref{meaning} we
illustrate some theoretical results and we carefully explain the
meaning of each parameter inside our model. In Section
\ref{estimation} we show and discuss some statistical tools in order
to estimate the model parameters from the data. In Section \ref{sim}
we provide a number of simulations in order to point out the
functioning of the model parameters and the ability of the proposed
estimation tools. Section \ref{real-data} deals with an application of
our model and instruments to a co-authorship network. Finally, in Section
\ref{conclusions} we give our conclusions and discuss some future
developments. The paper is enriched by an Appendix that contains a
theorem and its proof, and supplementary simulation results.

\section{Preliminaries}
\label{ass}

We assume new nodes sequentially join the network so that node $i$
represents the one that comes into the network at time step $i$. Each
node shows a finite number of features, that can be of different kinds
(key-words, main topics, spatial/geographical contexts, profile,
etc.), and different nodes can share the same features.  It is
worthwhile to note that we do not specify a priori the total number of
possible features but we allow the number of observed features to
increase along time. On its arrival, each new node links to some nodes
already present in the system.  Firstly, links are created according
to probabilities that depend on the number of common features
(multidimensional assortativity).  Then additional links can be
established by means of common neighbors, inducing the closure of some
triangles (triadic closure). We consider the connections as undirected,
non-breakable and we omit self-loops (i.e. edges of type $(i,i)$). In
particular, this means that connections are mutual or the direction is
naturally predefined (for instance, only citations from newer to older
nodes are possible). We denote the adjacency matrix (symmetric by
assumption) by $A$, so that $A_{i,j}=1$ when there exists a link
between nodes $i$ and $j$, $A_{i,j}=0$ otherwise. We set
$$
{\mathcal V}_j(i)=\{j'=1,\dots, i: A_{j,j'}=1\}
$$ 
to be the set of node $j$'s neighbors at time step $i$ 
(after the arrival of $i$).

\indent We denote by $F$ the binary bipartite network where each row
$F_i$ represents the features of node $i$: $F_{i,k}=1$ if node $i$ has
feature $k$, $F_{i,k}=0$ otherwise.  It represents the surrounding
context in which the nodes interact. We assume that each $F_i$ is
unchangeable during time.  We take $F$ left-ordered: this means that
in the first row the columns for which $F_{1,k}=1$ are grouped on the
left and hence, if the first node has $N_1$ features, then the columns of
$F$ with index $k\in\{1,\dots, N_1\}$ represent these features.  The
second node could have some features in common with the first node
(those corresponding to indices $k$ such that $k=1,\dots, N_1$ and
$F_{2,k}=1$) and some, say $N_2$, new features. The latter are grouped
on the right of the set for which $F_{1,k}=1$, i.e., the columns of
$F$ with index $k\in\{N_1+1,\dots, N_2\}$ represent the new features
brought by the second node. This grouping structure persists
throughout the matrix $F$ and we define  $L_n=\sum_{i=1}^n N_i$, i.e.
\begin{equation}\label{def-L}
\begin{split}
L_n=\hbox{overall number of different observed features 
for the first } n 
\hbox{ nodes}.
\end{split}
\end{equation}
Here is an example of a $F$ matrix with $n=3$ nodes:
\[
F=\left(
\begin{array}{cccccccccccc}
\cellcolor{gray!20}1 & \cellcolor{gray!20}1 & \cellcolor{gray!20}1 
& 0 & 0 & 0
& 0 & 0\\
1 & 0 & 1 & \cellcolor{gray!20}1 & \cellcolor{gray!20}1 & 0 & 0 & 0\\
0 & 1 & 1 & 1 & 0 & \cellcolor{gray!20}1 & \cellcolor{gray!20}1 &
\cellcolor{gray!20}1\\
\end{array}
\right).
\]
\noindent In gray we show the new features brought by each node (in
the example $N_1=3$, $N_2=2$, $N_3=3$ and so
$L_1=3,L_2=5,L_3=8$). Observe that, for every node $i$, the $i$-th row
contains 1 for all the columns with indices $k \in
\{L_{i-1}+1,\dots,L_i\}$ (they represent the new features brought by
$i$). Moreover, some elements of the columns with indices $k \in
\{1,\dots,L_{i-1}\}$ are also $1$ (features brought by previous nodes
adopted by node $i$).

\section{The model}
\label{model}

Fix $\alpha>0$, $\beta\in [0,1]$, $\delta\in [0,1]$, $p\in[0,1]$ and
let $\Phi: {\mathbb R} \to [0,1]$ be an increasing function.  The
dynamics is the following.  Node 1 arrives and shows $N_1$ features,
where $N_1$ is Poi$(\alpha)$-distributed (the symbol Poi$(\alpha)$
denotes the Poisson distribution with mean $\alpha$).  Then, for each
$i\geq 2$,
\begin{itemize}
\item {\bf Feature-structure dynamics:} Node
  $i$ arrives and shows a number of features as follows:
\begin{itemize}
\item Node $i$ exhibits some of the {\em ``old''} 
features  brought by the
  previous nodes $1,\dots, i-1$: 
more precisely, each feature $k\in
  \{1,\dots,L_{i-1}\}$ is, independently of the
  others, possessed by node $i$ 
 with probability 
(that we call ``inclusion-probability'')
\begin{equation}
P_{i}(k)=\delta \frac{1}{2}+(1-\delta)
\frac{\sum_{j=1}^{i-1} F_{j,k}}{i}\,,
\label{inclusion-prob}
\end{equation} 
where $F_{j,k}=1$ if node $j$ shows feature $k$ and
$F_{j,k}=0$ otherwise. 
\item Node $i$ also shows $N_{i}$ {\em ``new''} features, 
where $N_{i}$ is
  Poi$(\lambda_i)$-distributed with
\begin{equation}
\lambda_i=\frac{\alpha}{i^{1-\beta}}.
\label{lambda}
\end{equation}
\end{itemize}
($N_i$ is independent of $N_1,\dots, N_{i-1}$ and 
of the exhibited ``old'' features.)\\
\noindent The matrix element $F_{i,k}$ is set equal
to $1$ if node $i$ has feature $k$ and equal to zero
otherwise.
\item {\bf Network construction:} On its arrival,
  node $i$ determines a set ${\mathcal L}_i$ of neighbors among
  the nodes already present in the network (so that we set
  $A_{i,j}=A_{j,i}=1$ for each $j\in{\mathcal L}_i$) as follows:
\begin{itemize}
\item ({\em First phase}) First, a set ${\mathcal L}_i^*$ of neighbors
  of node $i$ is established on the basis of the features shown.  Each
  node $j$ already present in the network (i.e.  $1\leq j\leq i-1$) is
  included in ${\mathcal L}_i^*$, independently of the others, with
  probability $\Phi(S_{i,j})$, where
\begin{equation}\label{similarity}
S_{i,j}=
\sum_{k=1}^{L_i} F_{i,k}F_{j,k}.
\end{equation}
is the number of features that $i$ and $j$ have in common.
\item ({\em Second phase}) Then some extra neighbors 
  are added to ${\mathcal L}_i$ on the basis of
  common neighbors. For every node $j\in\{1,\dots, i-1\}\setminus
  {\mathcal L}_i^*$, each node $j'\in {\mathcal V}_j(i-1)\cap
  {\mathcal L}_i^*$ (i.e. each neighbor that $i$ and $j$ currently
  share) can induce, independently of the others, the additional link
  $(i,j)$ with probability $p$.
\end{itemize} 
\end{itemize}

\section{Meaning of the model parameters and some results}
\label{meaning}

We now illustrate the meaning of the model parameters and 
some mathematical results regarding our model.

\subsection{The parameters $\alpha$ and $\beta$} 

Let us start with $\alpha$ and $\beta$.  The main effect of $\beta$ is
to regulate the asymptotic behavior of the random variable $L_n$
defined in \eqref{def-L} as a function of $n$. In particular,
$\beta>0$ is the {\em power-law exponent} of $L_n$. The main effect of
$\alpha$ is the following: the larger $\alpha$, the larger the total
number of new features brought by a node. It is worth to note that
$\beta$ fits the asymptotic behavior of $L_n$ and then, separately,
$\alpha$ fits the number of new observed features per node. (In
Section \ref{simulations-F} we will discuss more deeply this fact.)
More precisely, we prove (see the Appendix) the following asymptotic
behaviors:
\begin{itemize}
\item[a)] for $\beta=0$, we have a logarithmic behavior 
of $L_n$, that is 
${L_n}/{\ln(n)}\stackrel{a.s.}\longrightarrow \alpha$; 
\item[b)] for $\beta\in (0,1]$, we obtain a power-law 
behavior, i.e. 
${L_n}/{n^{\beta}}\stackrel{a.s.}
\longrightarrow {\alpha}/{\beta}$.
\end{itemize}

\subsection{The parameter $\delta$}

The parameter $\delta$ tunes the phenomenon of {\em preferential
  attachment} in the spreading process of features among nodes.  The
value $\delta=0$ corresponds to the ``pure preferential attachment
case'': the larger the weight of a feature $k$ at time step $i-1$
(given by the numerator of the second element in
(\ref{inclusion-prob}), i.e., the total number of nodes that exhibit
it until time step $i-1$), the greater the probability that $k$ will
be shown by the future node $i$.  The value $\delta=1$ corresponds to
the ``pure i.i.d. case'' with inclusion probability equal to $1/2$: a
node includes each feature with probability $1/2$ independently of the
other nodes and the other features.  When $\delta\in (0,1)$, we have a
mixture of the two cases above: the smaller $\delta$, the more
significant is the role played by preferential attachment in the
transmission of the features to new nodes.

\subsection{The function $\Phi$ and the parameter $p$}

According to our model, when a new node enters the system, it links to
some (possibly zero, one, or more) old nodes by means of the two
phases network construction described in Section \ref{model}. In the
first phase, a new node $i$ connects itself to some of the old nodes
according to a probability depending on its own features and the ones
of the others.  The function $\Phi$ relates the ``first-phase
link-probability'' of $i$ to $j$ (with $1\leq j\leq i-1$) to their
``similarity'' $S_{i,j}$ defined by (\ref{similarity}). Since $\Phi$
is assumed to be an increasing function, a higher number of common
features between nodes $i$ and $j$ induces a larger probability for
them to connect (akin the principle of assortativity). For instance,
we can take the generalization of the logistic function, i.e. the
sigmoid function
\begin{equation}\label{sigmoid}
\Phi(s) = \frac{1}{1+e^{K(\theta-s)}}\qquad \hbox{with } 
K>0,\,\theta\in\R.
\end{equation}
\noindent The sigmoid function smoothly increases (from $0$ to $1$)
around a threshold $\theta$, while $K$ controls its smoothness: the
bigger $K$, the steeper the sigmoid. In particular, $K=1$ and
$\theta=0$ give the logistic function and, for $K\to +\infty$, $\Phi$
approaches to a step function equal to $1$ or $0$, if the variable $s$
is respectively greater or smaller than $\theta$ (in our model,
$\theta \geq 0$ means that the links are established deterministically
based on whether the two involved nodes have, or not, a similarity
bigger than $\theta$).  In the second phase, node $i$ can connect to
some of the nodes discarded in the first phase by means of common
neighbors (triadic closure).  The parameter $p$ regulates this
phenomenon. Indeed, it represents the probability that a node causes a
link between two of its neighbors.  More precisely, in the second
phase, the probability of having a link between node $i$ and a node
$j\in\{1,\dots, i-1\}\setminus {\mathcal L}_i^*$ is
$\left[1-(1-p)^{C_{i,j}}\right]$, where $C_{i,j}= \mbox{card}
\big({\mathcal V}_j(i-1)\cap{\mathcal L}_i^*\big)$ is the number of
common neighbors of $i$ and $j$ after the first phase. Consequently,
the ``second-phase link-probability'' between a pair of nodes
increases with respect to $p$ and the number of neighbors they share.
The case $p=0$ corresponds to the case in which the connections only
depend on the similarity among nodes. The case $p=1$ corresponds to
the case in which the connection is automatically established when
$C_{i,j}>0$.

\section{Estimation of the model parameters}
\label{estimation}

In this section we illustrate how to
estimate the model parameters from the data.\\

\indent Suppose we can observe the values of
$F_1,\dots, F_n$, i.e. $n$ rows of the matrix $F$, 
where $n$ is the
number of observed nodes. From the asymptotic 
behavior of $L_n$, we get 
that $\ln(L_n)/\ln(n)$ is a strongly consistent estimator
  for $\beta$, hence 
we can use the slope $\widehat\beta$ of the
  regression line in the log-log plot 
(of $L_n$ as a function of $n$) as
  an estimate for $\beta$.
 
After computing $\widehat\beta$, we can estimate
$\alpha$ as:
\begin{equation}\label{stima-alpha}
\begin{split}
\widehat\alpha&= 
\widehat\gamma\qquad\hbox{when } \widehat\beta=0\\
\widehat\alpha&={\widehat\beta}\,\widehat\gamma 
\qquad\hbox{when } 0<\widehat\beta\leq 1,
\end{split}
\end{equation}
where $\widehat\gamma$ is the slope of the regression 
line in the plot $\big(\ln(n), L_n\big)$ or 
in the plot $\big(n^{\widehat\beta}, L_n\big)$
according to whether $\widehat\beta=0$ or $\widehat\beta\in (0,1]$. 

\indent We can estimate $\delta$ by means of a 
maximum likelihood
procedure. For this purpose, we now give a general expression of 
the probability
of observing $F_1=f_1,\dots, F_n=f_n$ given the 
parameters $\alpha, \beta$, and $\delta$.

The first row $F_1$ is simply identified by 
$L_1=N_1$ and so
\begin{equation*}
\begin{split}
P(F_1=f_1)&=P(N_1=n_1=\mbox{card}\{k: f_{1,k}=1\})\\
&=\mbox{Poi}(\alpha)\{n_1\}=e^{-\alpha}\frac{\alpha^{n_1}}{n_1!}.
\end{split}
\end{equation*}
\noindent Then the second row is identified by the values 
$F_{2,k}$, with $k=1,\dots, L_1=N_1$, and by $N_2$, so that  
\begin{equation*}
\begin{split}
&P(F_2=f_2|F_1)=\\
&P(F_{2,k}=f_{2,k}\,\hbox{for } k=1,\dots,L_1,\, 
N_2=n_2=\mbox{card}\{k>L_1: f_{2,k}=1\} |F_1)
=\\
&\prod_{k=1}^{L_1} P_2(k)^{f_{2,k}}(1-P_2(k))^{1-f_{2,k}}
\times
\mbox{Poi}(\lambda_2)\{n_2\},
\end{split}
\end{equation*}
where $P_2(k)$ is defined in (\ref{inclusion-prob}) and 
$\lambda_2$ is defined in (\ref{lambda}). 
The general formula is 
\begin{equation*}
\begin{split}
&P(F_{i}=f_{i}|F_1,\dots,F_{i-1})=\\  
&P\left(F_{i,k}=f_{i,k}\,\hbox{for } 
k=1,\dots,L_{i-1},\,\right.
\\
&\quad\quad \left.N_{i}=n_{i}=\mbox{card}\{k>L_{i-1}: f_{i,k}=1\}
|F_1, \dots,F_{i-1}\right)
=\\ 
&\prod_{k=1}^{L_{i-1}} P_{i}(k)^{f_{i,k}}(1-P_i(k))^{1-f_{i,k}} 
\times
  \mbox{Poi}(\lambda_i)\{n_{i}\},
\end{split}
\end{equation*}
where $P_i(k)$ is defined in (\ref{inclusion-prob}) and 
$\lambda_i$ is defined in (\ref{lambda}). Thus, 
for $n$ nodes, we can write a formula for the 
probability of observing $F_1=f_1,\dots, F_n=f_n$:
\begin{equation} \label{likelihood}
\begin{split}
&P(F_1=f_1,\dots,F_n=f_n)=\\
&P(F_1=f_1)
\prod_{i=2}^{n} P(F_{i}=f_{i}|F_1,\dots,F_{i-1}).
\end{split}
\end{equation}
\noindent Therefore, we look for $\widehat\delta$ that 
maximizes the likelihood
function, i.e. the quantity $P(F_1=f_1,\dots,F_n=f_n)$ 
as a function
of $\delta$ (given the observed vectors $f_{i}$). 
Since some factors do
not depend on $\delta$, we can simplify 
the function to be maximized
as
\begin{equation}
\prod_{i=2}^{n}
\prod_{k=1}^{L_{i-1}} P_i(k)^{f_{i,k}}(1-P_i(k))^{1-f_{i,k}},
\end{equation}
or, equivalently, passing to the logarithm, as
\begin{equation}
\label{delta-loglike}
\sum_{i=2}^{n}
\sum_{k=1}^{L_{i-1}} f_{i,k}\ln\big(P_i(k)\big)+
(1-f_{i,k})\ln\big(1-P_i(k)\big).
\end{equation}
\\

\indent Now, suppose that we are also allowed to observe the adjacency
matrix $A=(A_{i,j})_{1\leq i,j\leq n}$ (meaning the final adjacency
matrix after the arrival of all the $n$ observed nodes and the
formation of all their links) and to know which are the links that
each of the $n$ observed nodes formed only by means of the previously
described first phase (i.e. only due to assortativity).  Denote by
$A'=(A'_{i,j})_{1\leq i,j\leq n}$ the adjacency matrix collecting
them. Then, if we decide to model the function $\Phi$ as in
(\ref{sigmoid}), we can choose $K$, $\theta$, and $p$, in order to fit
some properties of the observed matrices $A'$ and $A$. For instance,
if $\ell$ is the number of observed (undirected) links in matrix $A'$
(i.e. only due to the first phase of network construction) and
$$ f^*= \frac{ \hbox{observed number of linked (in $A'$) pairs of
    nodes with } s^*\, \hbox{ features in common}} {\hbox{observed
    number of pairs of nodes with } s^*\, \hbox{ features in common}},
$$ where $s^*$ is a fixed value that we choose, then we can determine $K>0$
and $\theta\in\R$ by solving (numerically) the following system of two
equations:
\begin{equation}\label{system-K-theta}
\begin{split}
&\Phi(s^*)=\left(1+e^{K(\theta - s^*)}\right)^{-1}=f^*
\\[3pt]
&E\left[\sum_{i,j: 2\leq i\leq n, 1\leq j\leq i-1} A'_{i,j}\right]
=\sum_{i=2}^{n}\sum_{j=1}^{i-1}\Phi\left(S_{i,j}\right)= 
\\
&
\sum_{i=2}^{n}\sum_{j=1}^{i-1}
\left(1+
e^{K(\theta - s^*)+K(s^*-\sum_{k=1}^{L_i} F_{i,k} F_{j,k})}
\right)^{-1}=\ell.
\end{split}
\end{equation}
By means of the first equation, we fit the probability that a pair of
nodes with $s^*$ features in common establishes a link (during the
first phase of network construction); while, by the second equation,
we set the expected number of links in $A'$ equal to the observed
$\ell$.  From the first equation, we get the quantity $K(\theta-s^*)$,
we then replace it in the second one in order to obtain $K$ and from
this we get $\theta$. Note that this is not a proper estimation
procedure, but rather a selection mechanism for $K$ and $\theta$ in
order to fit some observed properties of the network. After that, we
can estimate $p$ by means of a maximum likelihood procedure based on
the observed matrices.\\

\indent Some important remarks follow. If in the considered situation
the formation of links only occur according to the first phase
(i.e. as a result of the assortativity property), then we can set
$p=0$ as in this case the presence of closed triplets is only caused by
common features and the matrix $A$ coincides with $A'$. Then we have
no problem to implement the previous procedures for detecting all the
model parameters. When we have both phases of network construction
(i.e. $p>0$), the detection of $K,\theta$, and $p$ may generate some
problems since the available data are typically $F$ and $A$, while, in
order to implement the above procedure, we also need to observe $A'$.
When we cannot observe $A'$, we may try to reconstruct it from $A$ in
some consistent way, if it is possible for the considered application
\cite{lafond-neville}.  However, every empirical criterion used to
distinguish between the two different types of links (the ones due to
the first phase and the ones induced by the second phase), obviously
has some degree of arbitrariness and it can be hard to understand the
bias implied by it. An example of this problem can be found in
\cite{bramoulle2012jet} regarding a citation network. In the case no
suitable criterion is found, we may try to select $K,\theta$, and $p$
in such a way that some properties of the adjacency matrix generated
by the model are close to the observed one. Statistical procedures
that integrate out unobserved variables (in this case, $A'$) or
expectation-maximization (EM) algorithms are also possible and they
will be subject of future developments. Therefore, although
assortativity and triadic closure are theoretically well separated
concepts, in practice there are situations in which estimating them
singly is not a simple task. However, their combination is often
necessary in order to get models that produce realistic networks. The
simulation of the model with the observed matrix $F$ and $p=0$ can be
useful as a benchmark.

\section{Simulations}
\label{sim}

In this section, we present a number of simulations performed
following the dynamics for the features' selection and links' creation
described in Section \ref{model}. We simulated the outcome for feature
matrices and for unipartite networks of $1000$ nodes, on a sample of
$100$ realizations. Regarding the feature-selection dynamics, we
analyzed the resulting feature matrices (constructed as explained in
Section \ref{ass}) for different values of the model pa\-ra\-me\-ters
$\alpha$, $\beta$, and $\delta$, responsible respectively of the
number of new features per node, the asymptotic behavior of $L_n$
defined in \eqref{def-L}, and the phenomenon of preferential
attachment in the transmission of the features to new nodes. After
that, we simulated the network construction taking $\Phi$ as in
\eqref{sigmoid} and analyzed its properties for different values of
$\delta$, $K$, and $p$, while $\theta$ is determined according to a
certain number $\ell$ of (undirected) links due to the first phase of
the unipartite network construction.

\subsection{Simulations of the feature matrix and 
estimation of $\alpha,\,\beta$, and $\delta$}
\label{simulations-F}

As said before, parameter $\alpha$ is responsible for the number of
new features per node: the larger $\alpha$, the higher the number of
new features per node. Concerning this, it is very important to stress
that also the parameter $\beta$ affects the number of features per
node, but the idea is that we select first $\beta$, in order to fit
the asymptotic behavior of $L_n$, and then $\alpha$ in order to fit
the number of new features per node.\\

\indent In the first set of simulations we kept $\beta=0.5$ and
$\delta=0.1$ fixed and we built the feature matrix for different
values of $\alpha=3,\,8,\,13$. In Figure \ref{matrices-alpha} we can
see the shapes of the feature matrices (where colored points denote
non-zero values, i.e. $1$) for the three different values of
$\alpha$. It is immediate to see that the main difference among these
matrices concerns the number of features: the total number of features
is $185$ for $\alpha=3$, $533$ for $\alpha=8$, and $819$ for
$\alpha=13$.  Correspondingly, the mean number of new features per
node (averaged over $100$ realizations) is about $0.19$ for
$\alpha=3$, $0.49$ for $\alpha=8$, and $0.8$ for $\alpha=13$. The mean
number of (total) adopted features per node (averaged over $100$
realizations) is about $19.99$ for $\alpha=3$, $52.66$ for $\alpha=8$,
and $79.65$ for $\alpha=13$.  \\

\begin{figure}
\centering
\begin{subfigure}[b]{0.3\textwidth}
\includegraphics[width=\textwidth]{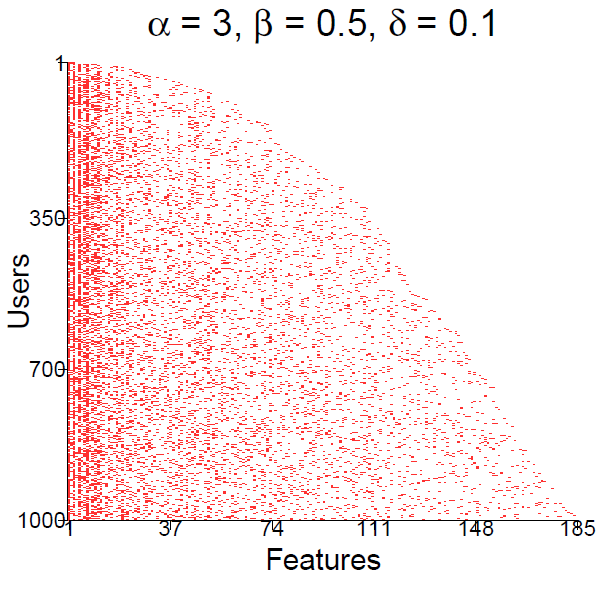}
\end{subfigure}
\begin{subfigure}[b]{0.3\textwidth}
\includegraphics[width=\textwidth]{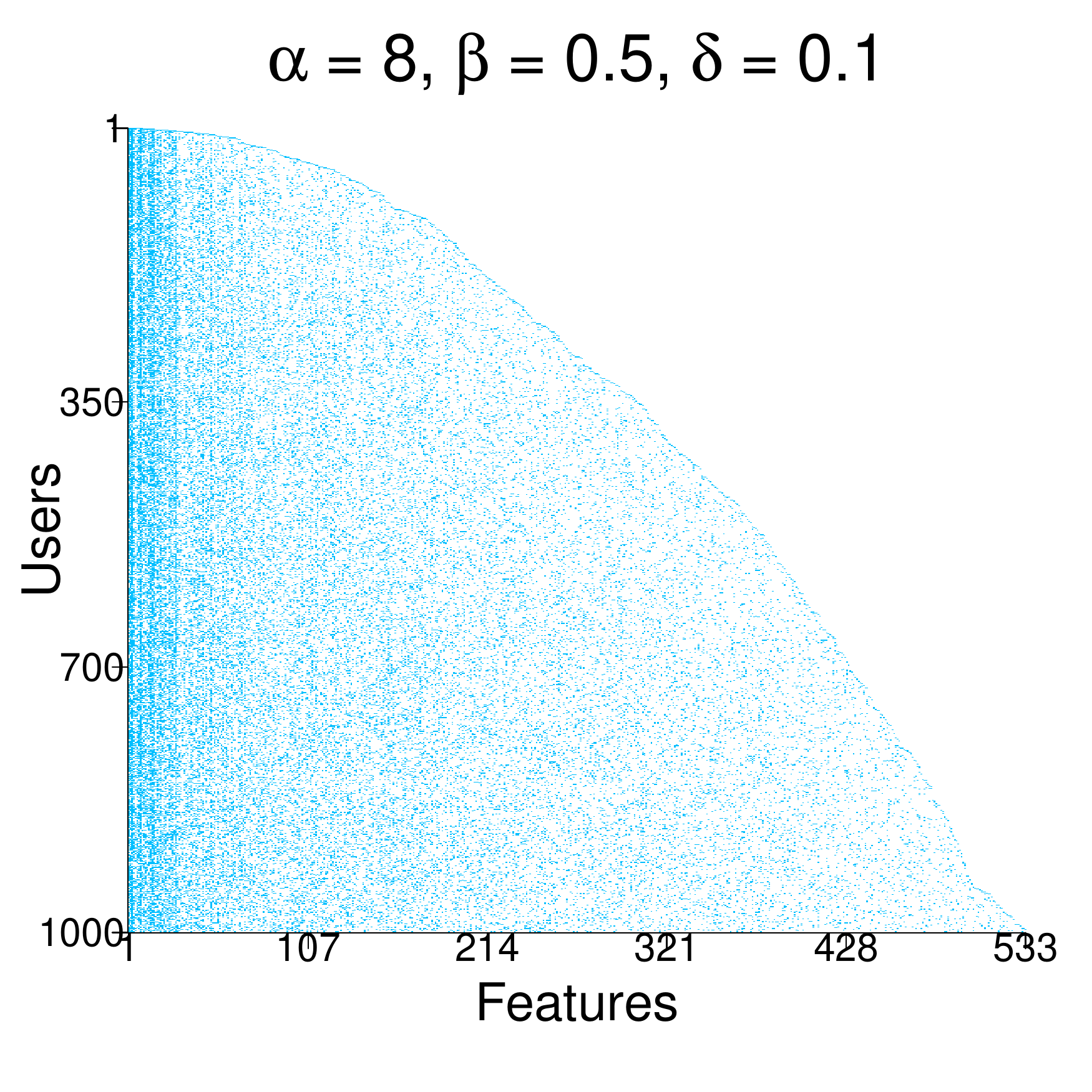}
\end{subfigure}
\begin{subfigure}[b]{0.3\textwidth}    
\includegraphics[width=\textwidth]{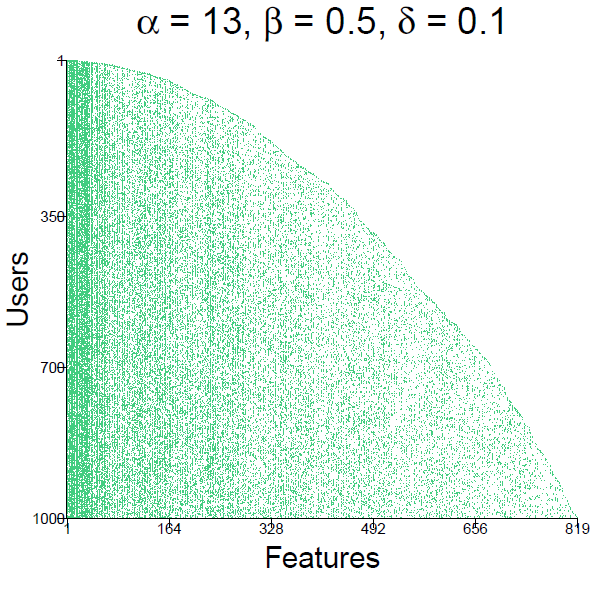}
\end{subfigure}
\caption{{\bf An example of features matrices} for $n=1000$,
  $\beta=0.5$, $\delta=0.1$, and different values of $\alpha: 3$ (left),
  $8$ (middle), $13$ (right).  Colored points denote $1$ and white
  points denote $0$.}
\label{matrices-alpha}
\end{figure}

In Figure \ref{featlen1} we show the estimates for the different
values of $\alpha$ (with $\beta=0.5$ and $\delta=0.1$ kept fixed).\\

\begin{figure}
\centering
\begin{subfigure}[b]{0.3\textwidth}
\includegraphics[width=\textwidth]{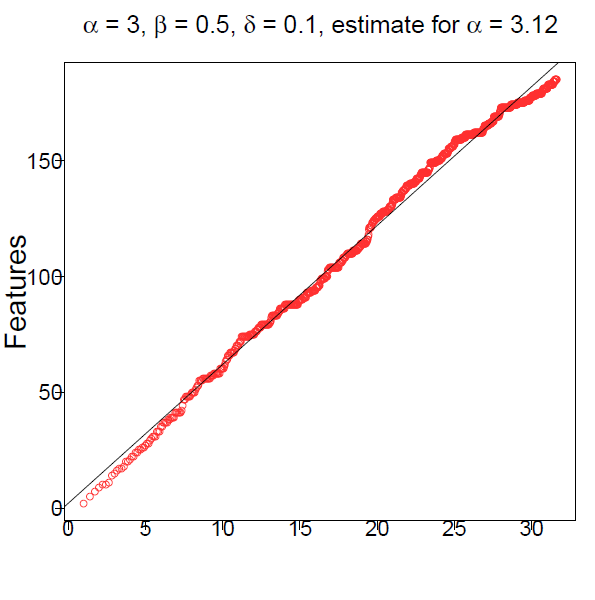}
\end{subfigure}
\begin{subfigure}[b]{0.3\textwidth}
\includegraphics[width=\textwidth]{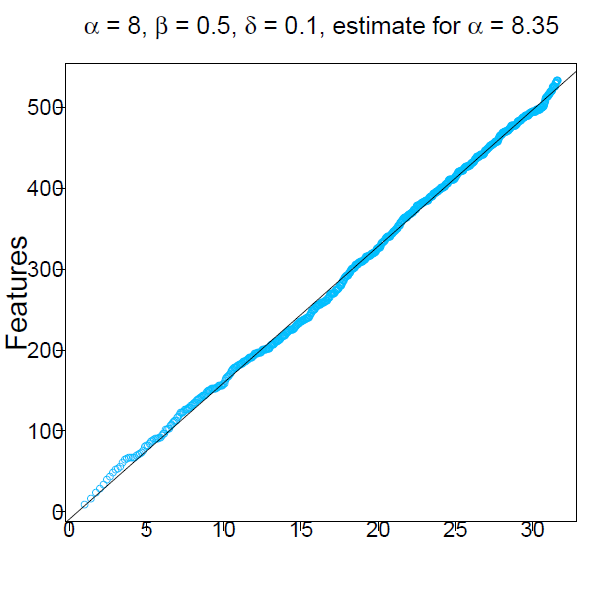}
\end{subfigure}
\begin{subfigure}[b]{0.3\textwidth}    
\includegraphics[width=\textwidth]{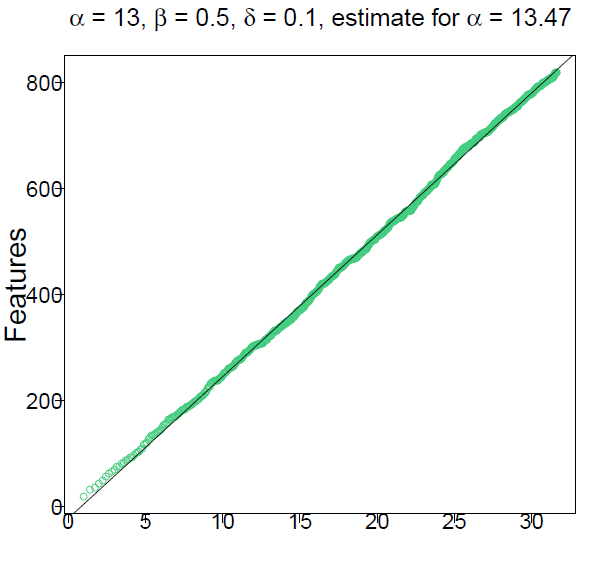}
\end{subfigure}
\caption{{\bf Estimates of $\alpha$} (when $\beta=0.5$ and $\delta=0.1$) 
  obtained as the slope of the regression line in the plot of $L_n$ as
  a function of $n^\beta$.  Different values of $\alpha: 3$ (left), $8$
  (middle), $13$ (right) are reported.  }
\label{featlen1}
\end{figure}

\indent Parameter $\beta$ controls the asymptotic behavior of
$L_n$. For this reason we plotted $L_n$ as a function of $n$ in a
log-log scale, results are reported in Figure \ref{featlen2}.  In
Figure \ref{featlen2} (a)-(b), we show the estimates for two different
values of $\beta$ ($\beta=0.75$ and $\beta=1$), with $\alpha=3$ and
$\delta=0.1$. In Figure \ref{featlen2} (c)-(d), we show the estimate
of $\beta$, for $\beta=0.5$ and $\beta=0.75$, but for a different
value of $\alpha$ ($\alpha=10$) in order to underline that $\alpha$
does not affect the power-law behavior of $L_n$ (obviously, the value
of the estimate can be more or less accurate for different values of
$\alpha$).  \\

\begin{figure}
\centering
\begin{subfigure}[b]{0.45\textwidth}
\includegraphics[width=\textwidth]{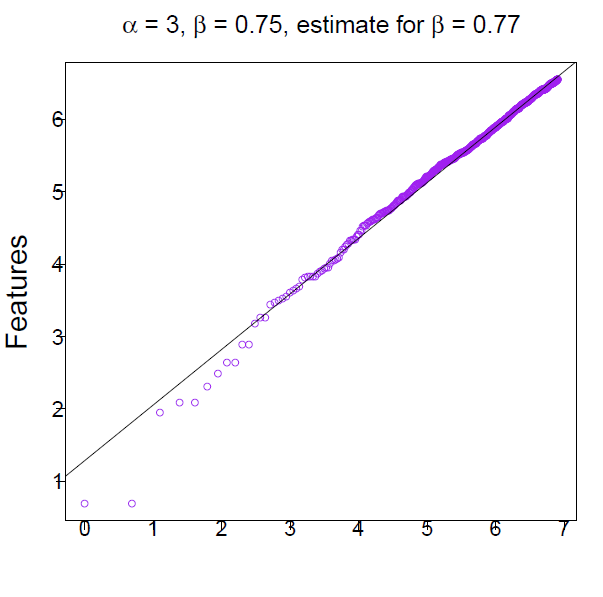}
\caption{}
\end{subfigure}
\begin{subfigure}[b]{0.45\textwidth}
\includegraphics[width=\textwidth]{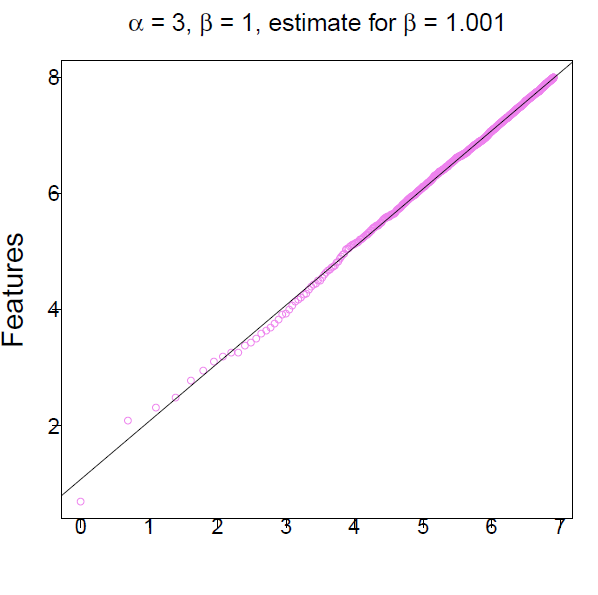}
\caption{}
\end{subfigure}

\begin{subfigure}[b]{0.45\textwidth}    
\includegraphics[width=\textwidth]{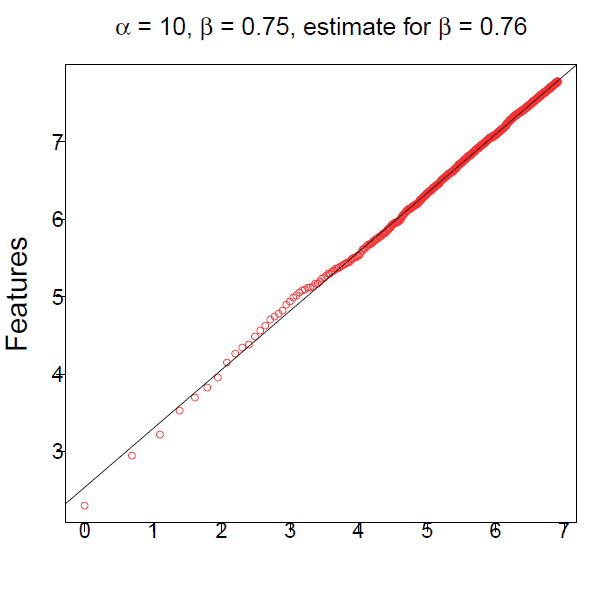}
\caption{}
\end{subfigure}
\begin{subfigure}[b]{0.45\textwidth}    
\includegraphics[width=\textwidth]{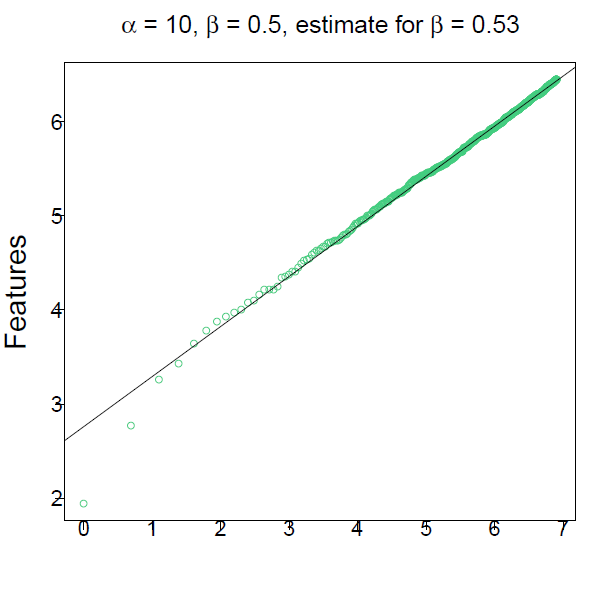}
\caption{}
\end{subfigure}
\caption{{\bf Estimates of $\beta$} obtained as the slope of the
  regression line in the log-log plot of $L_n$ as a function of $n$. 
  Different values of $\alpha$ and $\beta$ are reported:
  $\alpha=3,\,\beta=0.75$ (a), $\alpha=3,\,\beta=1$ (b),
  $\alpha=10,\,\beta=0.75$ (c),  and $\alpha=10,\,\beta=0.5$ (d).}
\label{featlen2}
\end{figure}

\indent Finally, parameter $\delta$ regulates the phenomenon of
preferential attachment: $\delta=0$ corresponds to the pure
preferential attachment case; while $\delta=1$ to the pure i.i.d case
with inclusion probability equal to $1/2$. The parameter $\delta$ is
estimated through the maximization of the likelihood function in
Equation \eqref{delta-loglike}.  Results for the estimated parameters
are reported in Table \ref{delta_tab}.\\
 
 \begin{table}[ht]
 \centering
{\scriptsize{
 \begin{tabular}{l |c c c c c c c c c c c }
 \hline
$\delta$ & 0 & 0.1 & 0.2 & 0.3 & 0.4 & 0.5 & 0.6 & 0.7 & 0.8 & 0.9 & 1\\ 
 $\hat{\delta}$ & 0.0002& 0.1002 & 0.2002 & 0.296 & 0.401 & 0.495 & 0.603& 
0.703 & 0.8 & 0.9 & 1.007\\
 \hline
 \end{tabular}
  }}
 \caption{\textbf{Estimates of $\delta$} computed as the maximum point
   $\widehat\delta$ of the likelihood function in formula
   \eqref{delta-loglike} with $\alpha=10$ and $\beta=0.5$.}
\label{delta_tab} 
 \end{table}
 
\indent In order to assess the accuracy of our estimation procedures,
we checked the Mean Squared Error (MSE) for all the three
parameters. More precisely, taking a sample of $R=100$ realizations,
we computed the quantities
\begin{equation*}
MSE_\alpha=\frac{1}{R}\sum_{r=1}^R({\widehat\alpha}_r-\alpha)^2,
\quad
MSE_\beta=\frac{1}{R}\sum_{r=1}^R({\widehat\beta}_r-\beta)^2,
\quad
MSE_\delta=\frac{1}{R}\sum_{r=1}^R({\widehat\delta}_r-\delta)^2,
\end{equation*}
where $\alpha,\, \beta,\, \delta$ are the values used to generate all
the $100$ realizations and ${\widehat\alpha}_r,\,{\widehat\beta}_r,\,
{\widehat\delta}_r$ are the estimated values associated with the
realization $r$.  For $\alpha=10,\,\beta =0.5,\, \delta=0.1$, we
obtained the following values:
$$
MSE_\alpha=1.18, \quad MSE_\beta = 0.0004, \quad MSE_\delta= 9\times 10^{-7}.
$$

\indent In Figure \ref{matrices-delta}, we show the shapes of the
feature matrices (where colored points denote non-zero values,
i.e. $1$) for different values of $\delta=0.1,\,0.5,\,0.95$ (two
different values of $\alpha=3,\,8$ and a fixed value of
$\beta=0.5$). Although the number of new features for each node is
comparable for different values of $\delta$ and a fixed value of
$\alpha$ (indeed, the parameter $\delta$ does not affect the number of
new features per node, but only the transmission of the old features
to the subsequent nodes), the number of old features selected by the
nodes depends on $\delta$: the more $\delta$ is near to zero, the more
the probability of showing an old feature depends on how many other
nodes selected it (preferential attachment). This fact is pointed out
by the ``full'' vertical lines, that are concentrated on the left-hand
side (since the preferential attachment phenomenon, the first features
are more successfully transmitted). For greater values of $\delta$,
the matrices become denser and they present a more uniform
distribution of the features among the nodes. The mean number of
(total) adopted features per node for $\alpha=3$ and $\delta$ equal to
$0.1,\,0.5$, and $0.95$ (averaged over $100$ realizations) is about
$19.99,\,44.24$, and $71.49$ respectively; while for $\alpha=8$ and
same values of $\delta$ it is approximately equal to $52.66,\,128.17$,
and $167.63$ respectively.\\

\indent In order to measure the ``uniformity'' of the distribution
of the features among nodes, we simply divided the total set of the
features into two subsets: $\{1,\dots, \lfloor L_n/2\rfloor\}$ and
$\{\lfloor L_n/2\rfloor +1, \dots, L_n\}$. For each feature, we
computed the mean number of nodes that adopted it (i.e. the total
number of nodes that adopted the considered feature divided by the
total number of nodes that could have adopted it). Then we computed
the mean value of these numbers over the two subsets and took the
difference between these two values.  For different values of $\alpha$
and $\delta$, Table \ref{tab_mean} contains the corresponding values
 (averaged over $100$ realizations) of these differences.
  It is clear that the smaller
the reported value, the more uniform is the distribution of the
features in the matrix. We can notice that for $\delta=0.1$ and
$\delta=0.5$ the obtained values are comparable (about $0.10$ and
$0.11$); while for $\delta=0.95$ we got a very small value.  \\

\begin{figure}
\centering
\begin{subfigure}[b]{0.3\textwidth}
\includegraphics[width=\textwidth]{Feat_mat_alfa3.png}
\end{subfigure}
\begin{subfigure}[b]{0.3\textwidth}
\includegraphics[width=\textwidth]{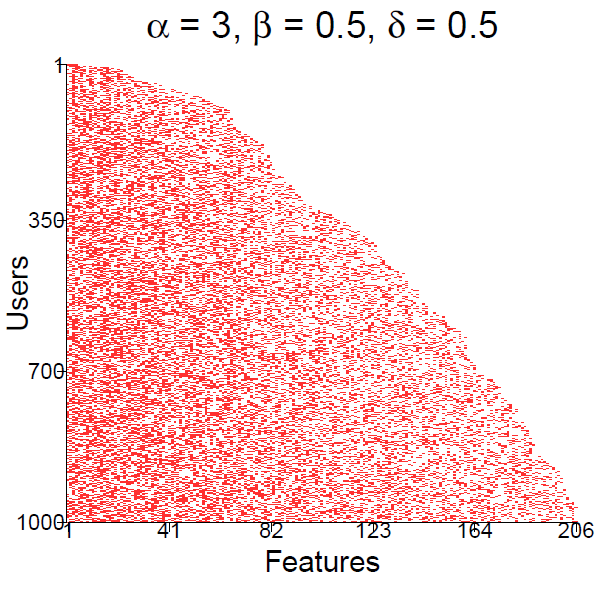}
\end{subfigure}
\begin{subfigure}[b]{0.3\textwidth}
\includegraphics[width=\textwidth]{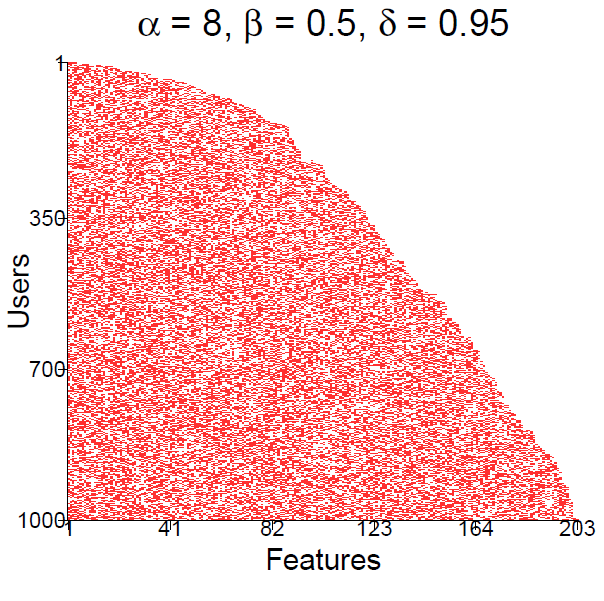}
\end{subfigure}

\begin{subfigure}[b]{0.3\textwidth}
\includegraphics[width=\textwidth]{Feat_mat_alfa8.pdf}
\end{subfigure}
\begin{subfigure}[b]{0.3\textwidth}
\includegraphics[width=\textwidth]{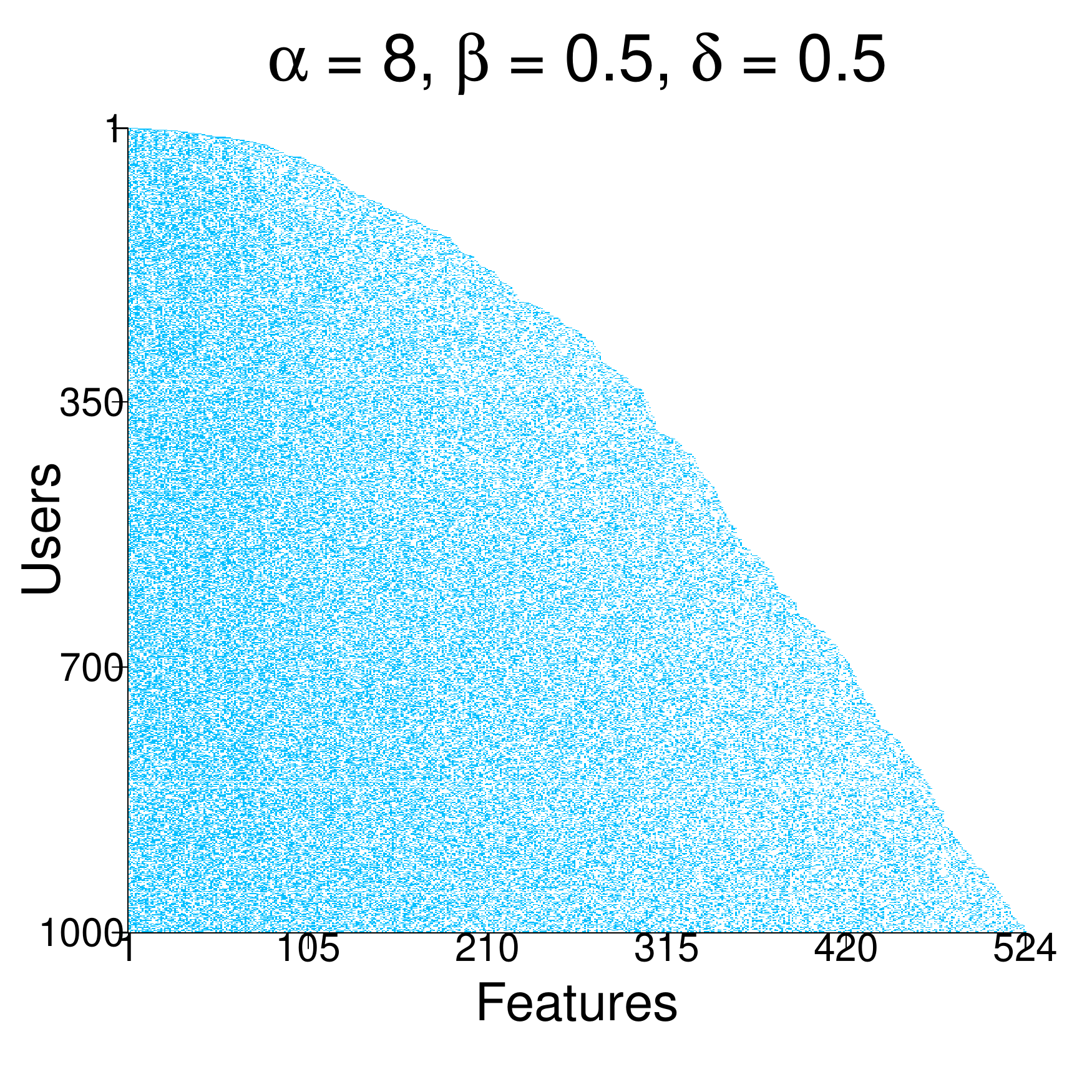}
\end{subfigure}
\begin{subfigure}[b]{0.3\textwidth}
\includegraphics[width=\textwidth]{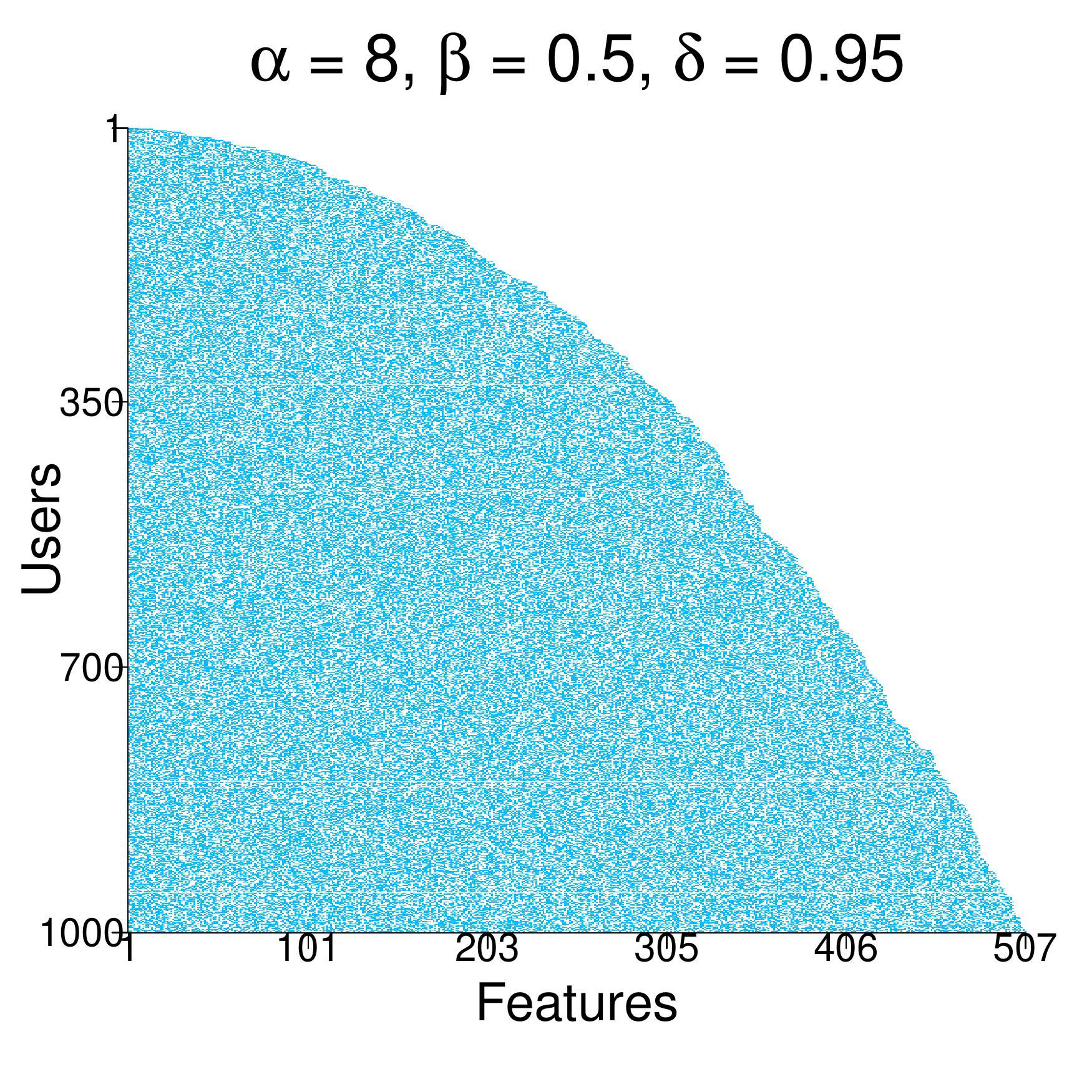}
\end{subfigure}
\caption{{\bf Examples of features matrices} for $n=1000$,
  $\beta=0.5$, different values of $\alpha: 3$ (top), $8$ (bottom) and
  different values of $\delta: 0.1$ (left), $0.5$ (middle), $0.95$
  (right).  Colored points denote $1$ and white points denote 
  $0$.}\label{matrices-delta}
\end{figure}

\begin{table}[ht]
 \centering
{\scriptsize{
 \begin{tabular}{l | c c c }

& $\delta\,=\,0.1$&  $\delta\,=\,0.5$  & $\delta\,=\,0.95$  \\ 
\hline
 $\alpha=3$ & 0.1005 &  0.1119 &  0.0099 \\

   $\alpha=8$& 0.1010&   0.1129 &  0.0097  \\
   \hline
 \end{tabular}
  }}
 \caption{\textbf{Measure of the ``uniformity'' of the feature matrix}
   defined as the difference (averaged over $100$ realizations)
   between the mean number of nodes per feature  for the first and the
   second half of the features' set. Considered parameters: $\alpha=3,\,8$,
   $\beta=0.5$ and $\delta=0.1,\,0.5,\,0.95$.
}
\label{tab_mean} 
\end{table}

\subsection{Simulations of the unipartite network and 
procedure in order to recover $K$ and $\theta$}

We performed the simulations of the unipartite network as follows.
Once a feature matrix $F$ is generated, links are created according to
the two phases of the link construction described in Section
\ref{model}, taking $\Phi$ as in \eqref{sigmoid}. We simulated the
network for $n=1\,000$ nodes on a sample of $100$ realizations.\\

\indent In the first set of experiments, we fixed a number $\ell$ of
links and we determined the value of $\theta$, for different values of
$K$, by solving (numerically) the equation
\begin{equation}\label{eq-theta}
\sum_{i=2}^{n}\sum_{j=1}^{i-1}
\left(1+
e^{K(\theta - \sum_{k=1}^{L_i} F_{i,k} F_{j,k})}
\right)^{-1}=\ell\,,
\end{equation}
in order to have the expected number of (undirected) links due to the
first phase of the unipartite network construction equal to the given
number $\ell$. Hence, we stu\-died the network structure as a function
of the parameters $K$ and $p$ (related to the link formation). In
particular, we recall that $p$ increases the triadic closure
phenomenon. We also considered different values of $\delta$, that
regulates the preferential attachment in the transmission of the
features and so influences the shape of the feature matrix $F$. In the
Appendix we report the results.\\

\indent With the second set of experiments, we studied the accuracy of
the procedure (\ref{system-K-theta}) used in order to recover $K$ and
$\theta$.  Hence, we fixed $\alpha=10$, $\beta=0.5$, $\delta=0.1$,
$K=1$, $\theta=10$, and $p=0$ (so that $A'=A$) and we generated a
sample of $R=100$ realizations of the network. We then applied the
procedure (\ref{system-K-theta}) to each realization $r$ (with
$s^*=10$\footnote{We also consider different values for $s^*$ and we
  obtain similar results.}) in order to get the corresponding values
${\widehat K}_{r}$ and ${\widehat\theta}_r$. We found:
\begin{equation*}
\begin{split}
&\frac{1}{R}\sum_{r=1}^R{\widehat K}_r=1.000462,
\qquad
MSE_K=\frac{1}{R}\sum_{r=1}^R({\widehat K}_r-K)^2=0.00415,\\
&\frac{1}{R}\sum_{r=1}^R{\widehat\theta}_r=9.998843,
\qquad
MSE_\theta=\frac{1}{R}\sum_{r=1}^R({\widehat\theta}_r-\theta)^2=0.00010.
\end{split}
\end{equation*}

\section{Application to a co-authorship network} 
\label{real-data}

We downloaded bibliographic information of papers and
preprints found in the IEEE Xplore database \cite{explore}.  In this
dataset a link is taken as the co-authorship of a paper between two or
more authors and the contexts of the papers are given by $2$-grams
(pairs of sequential words in the title or abstract). We selected the
papers using search terms related to the specific research area of
autonomous cars (also called connected cars).  

\subsection{Description of the dataset}

We downloaded (on Aug. 7, 2014) all papers in the IEEE preprint and
paper archive using $17$ specific search terms: `Lane Departure
Warning', `Lane Keeping Assist', `Blindspot Detection', `Rear
Collision Warning', `Front Distance Warning', `Autonomous Emergency
Braking', `Pedestrian Detection', `Traffic Jam Assist', `Adaptive
Cruise Control', `Automatic Lane Change', `Traffic Sign Recognition',
`Semi-Autonomous Parking', `Remote Parking', `Driver Distraction
Monitor', `V2V or V2I or V2X', `Co-Operative Driving', `Telematics \&
Vehicles', and `Night vision'.  The IEEE archive returned all the
papers in their database that contain these terms in the title or
abstract, and we downloaded the bibliographic records for all returned
papers including the authors, title, abstract, and the date on which
the paper was added to the database.  This download yielded $6\,129$
distinct papers with a complete bibliographic record and at least two
authors.  While these search terms can not be expected to yield all
papers related to automated car research, we expect to have found a
relatively broad panel of related papers.

\subsection{Analysis of the feature-structure}

The feature matrix was built by extracting all $2$-grams (pairs of
words) appearing in either the title or abstract of a paper.  The text
was converted to lowercase, removing all punctuation (with the
exception of the `/' and `.' characters) and multi-spaces, and split
into individual sentences.  The $2$-grams occurring in any sentence in
the title or abstract were labeled as features of the paper. In order
to remove spurious $2$-grams (e.g. `this paper' often occurs in the
abstract, but it is not relevant to connected cars), we exclude any
$2$-grams containing any of the words: `the', `a', `of', `and', `to',
`is', `for', `in', `an', `with', `by', `from', `on', `or', `that',
`at', `be', `which', `are', `as', `one', `may', `it', `and/or', `if',
`via', `can', `when', `we', `his', `her', `their', `this', `our',
`into', `has', `have', `only', `also', `do', `does', `presents',
`paper', `doesn't', and `not'.  This approach gave $155\,897$ distinct
{\em $2$-grams} (features) for a total of $6\,129$ {\em papers}
(nodes). We ordered the papers chronologically based on their entry
date into the IEEE database (which we expect to be a good proxy for
their publication date). The $2$-grams were ordered in terms of their
first appearance in a paper (as described in Section \ref{ass}).  \\

\indent Having extracted the set of the $2$-grams contained in each
paper, we constructed the feature-matrix $F$, with $F_{ik}=1$ if paper
$i$ contains the $2$-gram $k$ and $F_{ik}=0$ otherwise.  The resulting
matrix $F$ is shown in Fig. \ref{features_matrix_data}(a), with
non-zero values of $F$ indicated by colored points.  We also simulated
the feature-matrix for a smaller network of $1000$ nodes taking the
parameters equal to the corresponding estimated values (see
Fig. \ref{features_matrix_data}(b)). The number of features obtained
in the simulation is $28\,664$, which is consistent with the observed
matrix.

\begin{figure}
\centering
\begin{subfigure}[b]{0.59\textwidth}
\includegraphics[width=\textwidth]{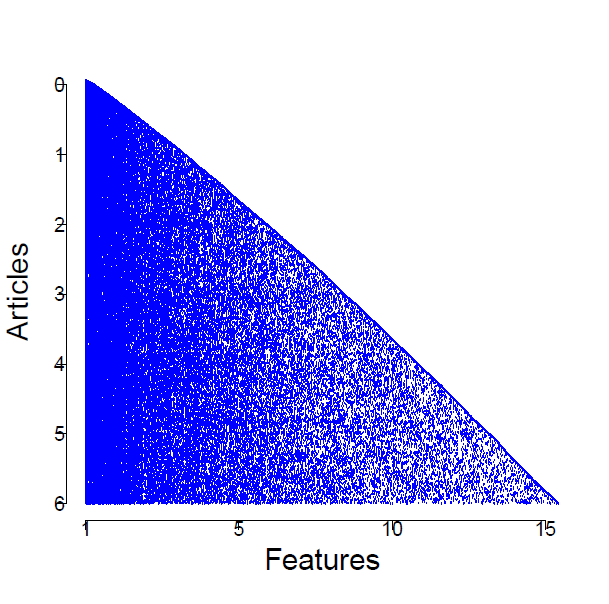}
\caption{}
\end{subfigure}
\begin{subfigure}[b]{0.40\textwidth}
\includegraphics[width=0.6\textwidth]{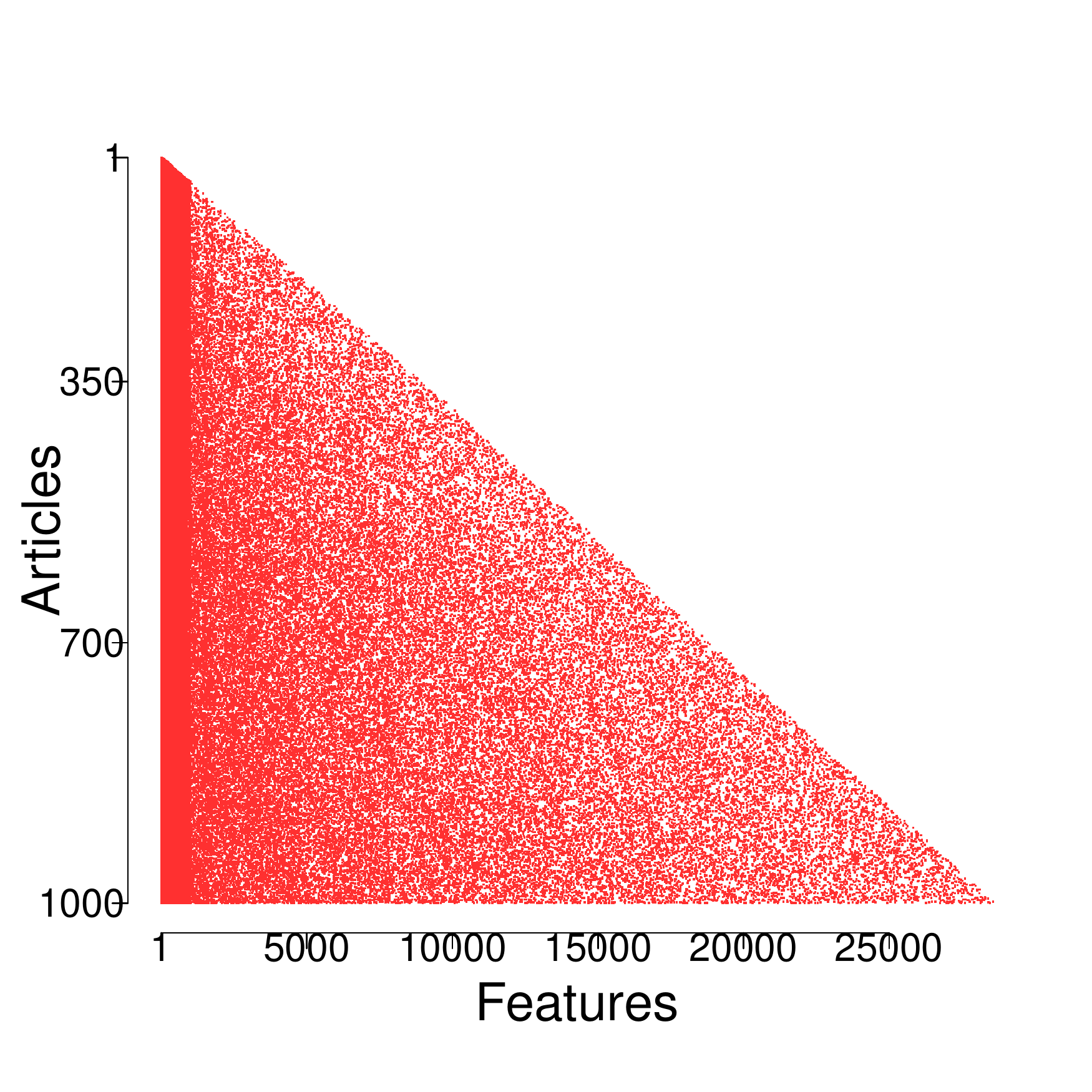}
\caption{}
\end{subfigure}
\caption{(a) Feature-matrix associated to the dataset. Dimensions:
  $6\,129$ nodes (papers) $\times$ $155\,897$ features
  ($2$-grams). Colored points denote $1$ and white points denote $0$. (b)
  Feature matrix for $1000$ nodes, obtained by simulation of the model
  with $\alpha=\widehat\alpha=32.28$, $\beta=\widehat\beta=0.98$, and
  $\delta=\widehat\delta=0.0057$. Colored points denote $1$ and white
  points denote $0$. The total number of features is $28\,664$, which
  is consistent with the observed matrix.}
\label{features_matrix_data}
\end{figure}

\begin{figure}
\centering
\begin{subfigure}[b]{0.40\textwidth}
\includegraphics[width=\textwidth]{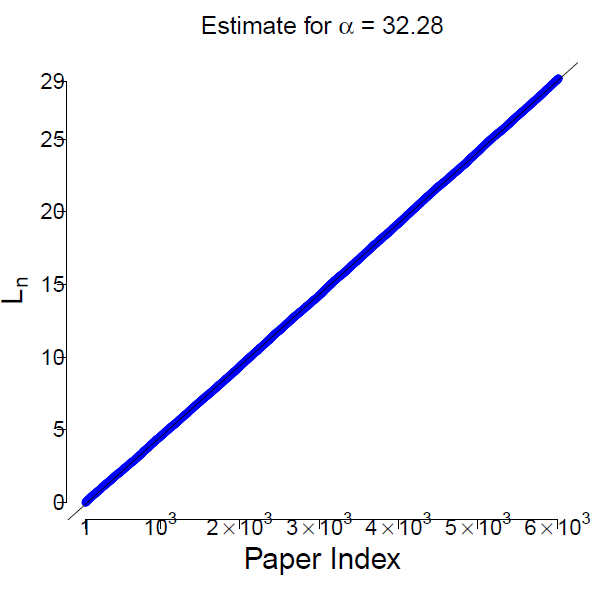}
\end{subfigure}
\begin{subfigure}[b]{0.40\textwidth}
\includegraphics[width=\textwidth]{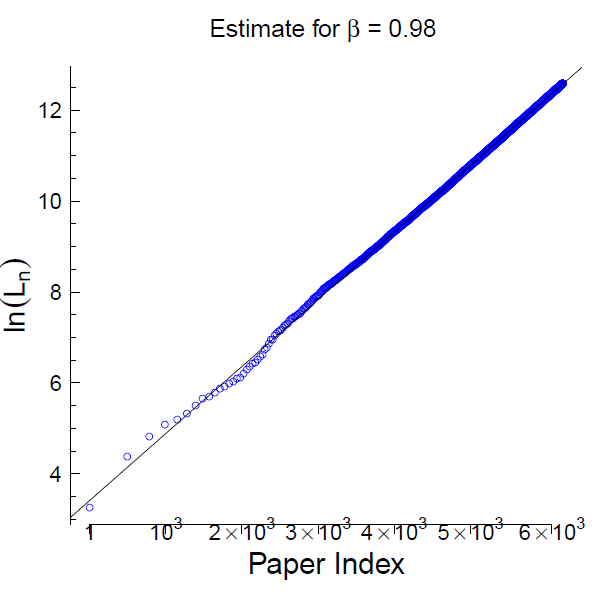}
\end{subfigure}
\caption{Estimated values of the model parameters $\alpha$ and
  $\beta$.}
\label{estimates_data}
\end{figure}

The growth of the cumulative count $L_n$ of the distinct $2$-grams
(the number of distinct $2$-grams seen until the $n^{th}$ paper
included, as described in Section \ref{ass}) is shown in
Fig. \ref{estimates_data}(b) in a log-log scale and it shows a clear
power-law behavior, with estimated parameter $\widehat{\beta}= 0.98$
(that corresponds to the estimated value of the model parameter
$\beta$). Regarding the model parameter $\alpha$, we get the estimated
value $\widehat{\alpha}=32.28$ and in Fig. \ref{estimates_data}(a) we
show the corresponding fit plotting the cumulative count $L_n$ of the
$2$-grams as a function of $n^{\widehat\beta}$. Finally, the estimated
value for the parameter $\delta$ is $\widehat\delta=0.0057$. As we can
see, this last value is very small and so we can conclude that the
preferential attachment rule in the transmission of the features plays
an important role.

\subsection{Analysis of the unipartite network}

Our dataset includes $6\,129$ papers for a total of $13\,581$
  distinct author names. The considered unipartite network is
constructed taking the papers as nodes and {\em drawing a link between two
nodes if they share at least one author}. We harmonized the author
names across different papers by ensuring that the authors' last names
are always found in the same position and removed any stray
punctuation in the names.  No further disambiguation was performed,
meaning that authors who may use their full names in some papers but
only their initials in other papers will be treated as distinct. For
example, the names ``J. J. Anaya'' and ``Jose Javier Anaya'' are
treated as distinct authors in our dataset, while it is possible that
these distinct names refer to the same person. A full disambiguation
of author names is computationally difficult \cite{disambiguation},
and beyond the scope of this paper. This approach gave a unipartite
network with $19\,065$ links that involve $4\,712$ nodes in the
network. This means that there is a set of $1\,417$ isolated nodes,
where a paper has two or more authors that are not listed on any other
paper in the dataset. However, we decided to also consider these nodes
in our analysis since we included them in the features matrix as nodes
that can potentially link to other nodes.  \\

\indent The distribution of the $2$-grams (the features) in common
between two papers (the nodes) given the presence or the absence of at
least one shared author (i.e. given the presence or the absence of a
link between them) is plotted in Figure \ref{distributions}(a). The
curve with (red) triangles is the distribution of the number of $2$-grams
shared by two papers given they have at least one co-author. More
precisely, for each value on the $x$-axis, we have on the $y$-axis the
fraction
\begin{equation}\label{fraction1}
\frac{
\hbox{num. of pairs of papers with } x\, \hbox{2-grams in common and at 
least 1 shared author}
}
{\hbox{num. of pairs of papers with at least 1 shared author}}.
\end{equation}
The curve with (green) stars represents the distribution of the number
of $2$-grams shared by two papers given they have no authors in
common, i.e. it is given by the same formula as (\ref{fraction1}) but
with pairs of papers without shared authors. As we can see, there is a
higher probability of common $2$-grams when there are shared authors. \\

\indent The fraction of pairs of papers with $x$ $2$-grams in common
that have at least one shared author is plotted in Figure
\ref{distributions}(b) by the curve with (red) triangles. More
precisely, for each value on the $x$-axis, we have on the $y$-axis the
fraction
\begin{equation}\label{fraction2}
\frac{
\hbox{num. of pairs of papers with } x\, \hbox{2-grams in common and at 
least 1 shared author}
}
{\hbox{num. of pairs of papers with } x\, \hbox{2-grams in common}}.
\end{equation}
As we can see, the plotted fraction increases with the number of
features in common. \\

\begin{figure}
\centering
\begin{subfigure}[b]{0.45\textwidth}
\includegraphics[width=\textwidth]{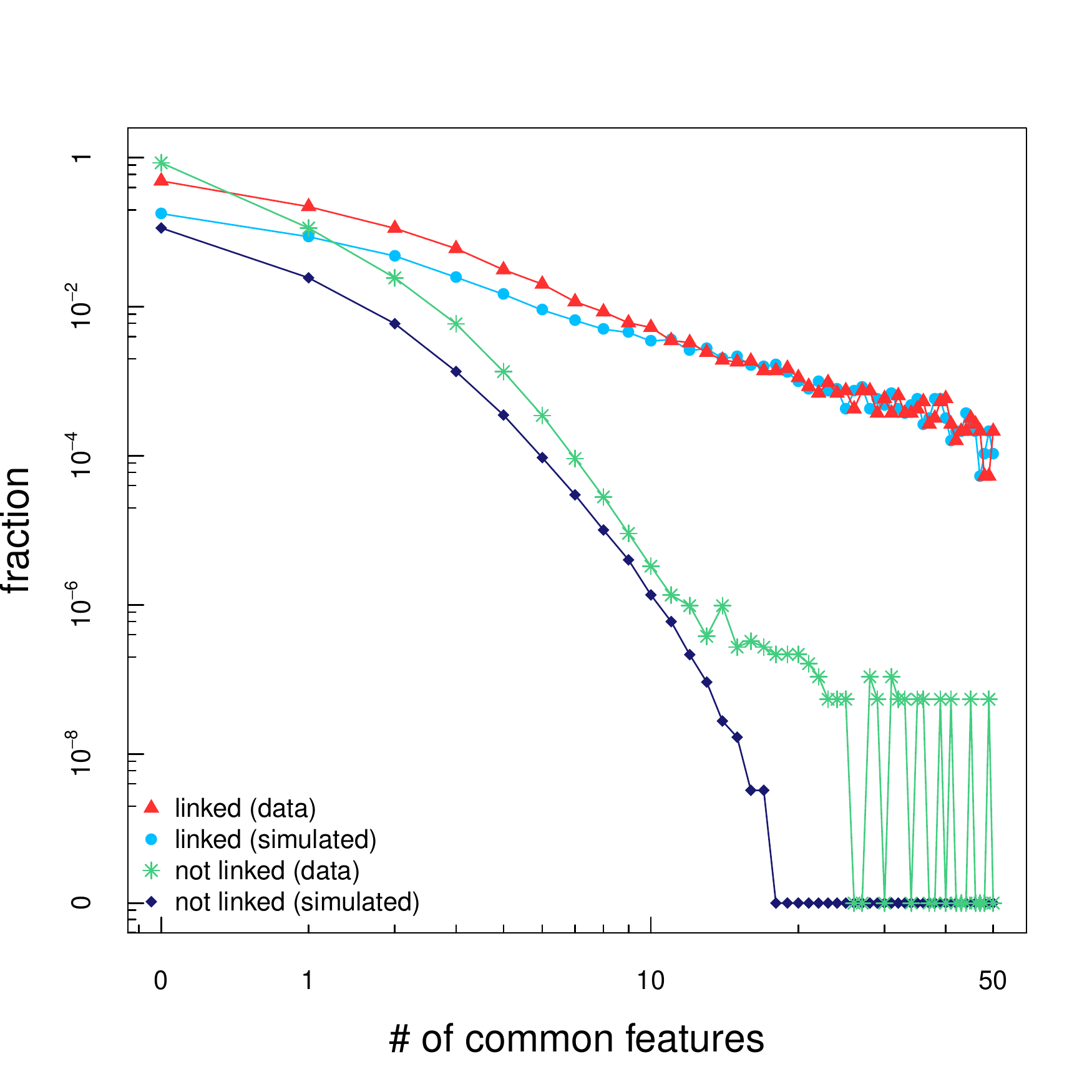}
\caption{}
\end{subfigure}
\begin{subfigure}[b]{0.45\textwidth}
\includegraphics[width=\textwidth]{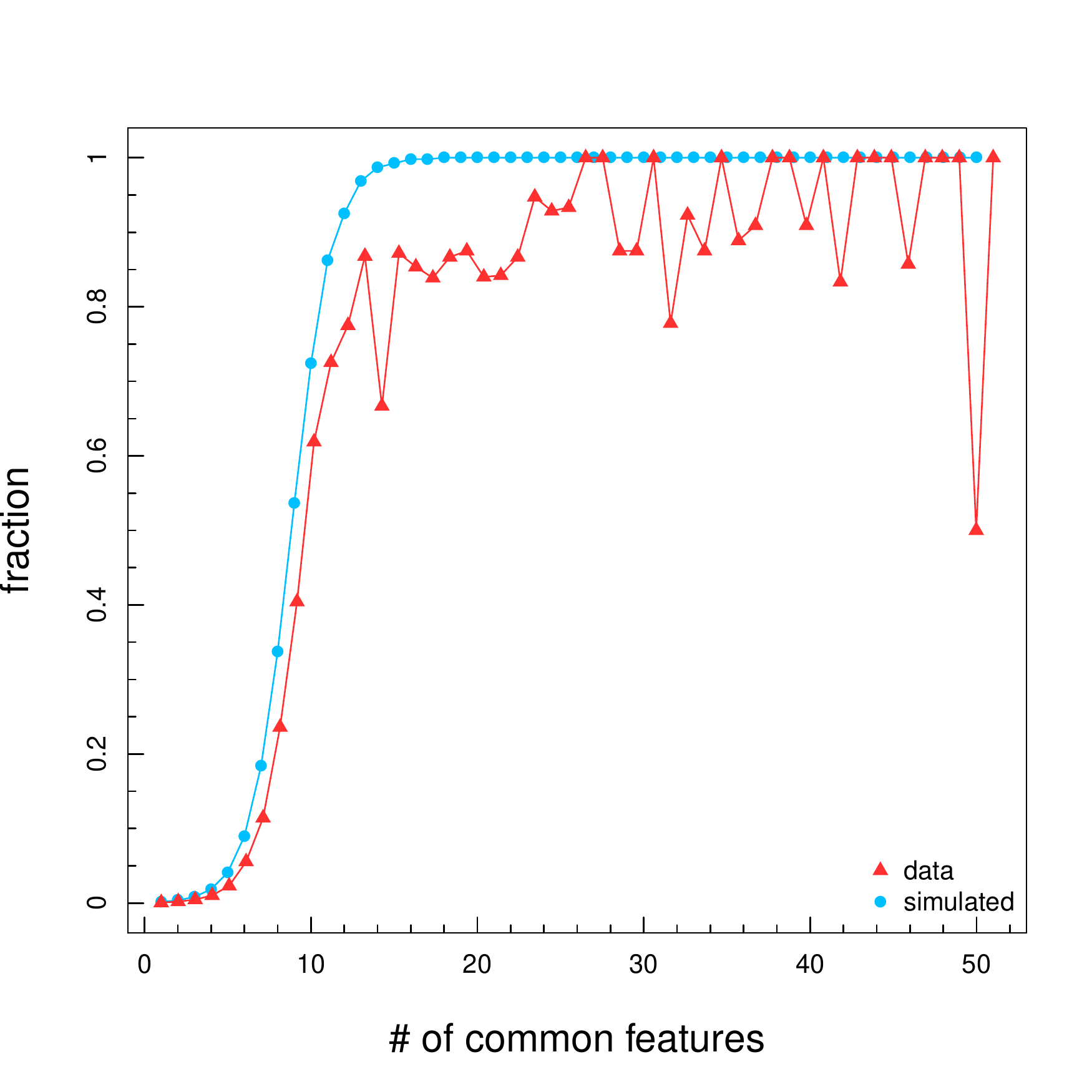}
\caption{}
\end{subfigure}

\caption{(a) Distribution of the $2$-grams (features) in common
  between two papers (nodes) given the presence (red triangles for the
  real data and light blue circles for the simulated ones) or the
  absence (green stars for the real data and dark blue squares for the
  simulated ones) of at least one co-author. (b) Fraction of pairs of
  papers with $x$ $2$-grams in common that have at least one co-author,
  for the real data case (red triangles) and for the simulated one
  (light blue circles).}
\label{distributions}
\end{figure}

\indent The network is composed of $586$ connected components with at
least one edge and $1\,417$ isolated nodes (a total of $2\,003$
components). The largest connected component has $2\,776$ nodes and
$16\,108$ links, so about the $45\%$ of the nodes can reach each other
in the largest connected component and it includes about the $84\%$ of
the links. The diameter (i.e. the maximum distance between nodes) of
the largest connected component is $23$. The other $585$ connected
components (disconnected from the largest component but still having
at least one edge) globally contain $1\,936$ nodes, and over $90\%$ of
the components (containing over $75\%$ of the nodes outside of the
largest connected component) contain $7$ or fewer nodes. Hence the
percentage of reachable pairs (denoted by $RP$ in the remainder of the
paper) of nodes in the network is about $20.51\%$.  \\

\indent We decided to first use the model with $p=0$ in order to have
a benchmark and then try to guess a good value for $p$. Taking $p=0$,
we set $A'=A$ (i.e. links are only formed by means of the first phase)
and we applied the procedure (\ref{system-K-theta}) to the observed
feature-matrix $F$ with $s^*=10$ (the corresponding value for $f^*$ is
$0.725$) and $\ell=19\,065$ in order to detect $K$ and $\theta$: we
found $\widehat{K}=0.8228$ and $\widehat\theta = 8.8201$.  We then
generate a sample of $100$ realizations of the network by simulating
the model starting from the observed matrix $F$ and with $p=0$, $K=
{\widehat K} = 0.8228$, and $\theta = {\widehat\theta} = 8.8201$.  We
obtained a network structure very different from the observed one (for
instance, $RP=99\%$). This can be obviously explained by the fact that
we set $p=0$ (benchmark case), while a value of $p$ strictly greater
than $0$ is guessable. Indeed, an author with $m\geq 3$ papers
automatically guarantees a minimum of $m \choose 3$ triangles. Setting
$p=0.7$ and generating a sample of $100$ realizations of the network
by simulating the model starting from the observed matrix
$F$ \footnote{In this case we took into account that $A'$ is different
  from $A$, and so the parameters $K$ and $\theta$ used for the
  simulations were recovered by applying the procedure
  (\ref{system-K-theta}) to the observed feature-matrix $F$ with a
  smaller $\ell$ (that corresponds to the expected number of links
  formed during the first phase). We set $\ell=4\,000$ in order to
  have an averaged total number of links around the observed one. We
  found ${\widehat K} = 1.019574$ and $ \widehat\theta = 9.047858$.},
we succeeded to capture a value for $RP$ very near to the observed
one, i.e. $RP=19.61\%$ (this value is an average over the $100$
realizations). Moreover, we obtained that the largest connected
component contains on average $2\,689.16$ nodes, again a value near to
the observed one.  Finally, Figure \ref{distributions}(a) contains the
distribution of the features in common between two nodes given the
presence (light blue circles) or the absence (dark blue squares) of a
link between them and Figure \ref{distributions}(b) depicts the
fraction of pairs of nodes with $x$ features in common that are
linked. Although the curves related to real data are obviously more
irregular, the curves generated by simulations properly fit to the
observed ones.

\section{Conclusions and discussion on some variants of the model}
\label{conclusions}

In this paper, we presented a new network model, where each node is
characterized by a number of features and the probability of a link
between two nodes depends on the number of features and neighbors they
share, so that it includes two of the most observed phenomena in
complex systems: assortativity, i.e. the prevalence of network-links
between nodes that are si\-mi\-lar to each other in some sense, and
triadic closure, meant as the high probability of having a link
between a pair of nodes due to common neighbors. The bipartite network
of nodes and features grows according to a stochastic dynamics that
depends on three parameters respectively regulating the preferential
attachment in the transmission of the features to the nodes, the
number of new features per node, and the power-law behavior of the
total number of observed features.  We provide theoretical results and
statistical tools for the estimation of the model parameters involved
in the feature-structure dynamics. From the observation of the
feature-matrix, we completely determine the parameters, $\alpha,
\beta, \delta$, that regulate its evolution. We provide a procedure
for recovering the two parameters, $K,\, \theta$, of the function
$\Phi$, which relates the link probability between two nodes to their
similarity in terms of common features, and the parameter $p$ which
tunes triadic closure. However, as discussed in Section
\ref{estimation}, for this last point, we need to know which are the
links formed by assortativity and those formed by triadic closure, but
often they are not easily distinguishable. Therefore we aim in the
future to evaluate more sophisticated estimation techniques for this
issue. Nevertheless, as shown in Section \ref{real-data}, we can still
exploit the proposed procedure in order to guess a good combination of
these parameters.\\

\indent The originality and the merit of our model mainly lie in the
double temporal dynamics (one for the feature-structure and one for
the network of nodes), but also in the attention given to both
assortativity and triadic closure mechanisms. We underline that,
differently from other models in the literature, we do not require to
specify a priori the values of some hyperparameters, such as the total
number of possible features (avoiding some selection problems
discussed in \cite{kri}). In the future, we aim at improving our model
in order to make it suitable for other kind of networks, e.g. real
social networks (such as friendship networks). In particular, the
following variations are possible:
\begin{itemize}
\item{\em Normalizing the number of common features:} We can vary the
  model by replacing the factor $F_{i,k}F_{j,k}$ in formula
  (\ref{similarity}) with
$$
\frac{F_{i,k}F_{j,k}}{\sum_{j'=1}^{i-1} F_{j',k}}, \, 
\forall (i,j) \mbox{ with } 1\leq j\leq i-1,
$$ 
\noindent so that the contribution of a common feature $k$ is smaller
when the number of nodes with $k$ as a feature is larger.

\item{\em Weighted bipartite matrices:} We can modify the model by
  replacing in the inclusion-probability and in the link-probability
  the binary random number $F_{i,k}$ by a random weight $W_{i,k}$ of
  the form $W_{i,k}=F_{i,k}Y_{i,k}/(\sum_{k=1}^{L_i} F_{i,k}Y_{i,k})$,
  where $Y_{i,k}$ are i.i.d.  strictly positive random variables. (By
  convention, we set $0/0=0$.)  Hence, we have
$$ 
W_{i,k}\in [0,1]
\quad\mbox{and}\quad 
\sum_{k=1}^{L_i} W_{i,k}=1
$$ so that $W_{i,k}$ represents the weight percentage given to feature
$k$ by node $i$. Therefore, the preferential attachment in the
inclusion-probability becomes a ``weighted preferential attachment'',
in the sense that it depends on the total weight given to feature $k$
by the previous nodes, and the link-probability depends on the weights
associated to the common features.

\item {\em Changeable links:} For some real situations, we need to
  consider also the case in which the links among nodes can change
  along time.  For instance, we can combine a link-formation model and
  a link-dissolution model as in \cite{kri2014}. See also
  \cite{hs2003} for node exit.

\item{\em Exit of some features and social influence of links on
  features:} We can modify the evolution of the feature-structure by
  accounting for the fact that at each time step $j$ (after the
  arrival of the node $j$) some features can become ``obsolete'' and
  so for such a feature $k$ we will have $F_{i, k}=0$ for all $i\geq
  j+1$. Moreover, a node could change some features under the
  influence of its ``friends'' (i.e. neighbors) \cite{hs2003}.  Hence,
  we can introduce a sequence $(F^{(i)})_i$ of bipartite matrices such
  that each $F^{(i)}$ provides the features before the arrival of node
  $i+1$, so that in the inclusion-probabilities and in the
  link-probabilities for node $i+1$, the matrix $F$ is replaced by
  $F^{(i)}$.

\item{\em Different dynamics for triadic closure:} We can change the
  second phase of our model by means of different policies for the
  selection of additional neighbors of a node $i$ among the neighbors
  of $i$'s neighbors.  Indeed, in this paper we consider a binomial
  model according to which each common neighbor of a pair $(i,j)$ of
  not-linked nodes gives, independently of the others, a probability
  $p$ of inducing a link between $i$ and $j$. A possible alternative
  is that, with probability $p$, an additional link for a certain node
  is formed by the selection (uniformly at random) of a node among the
  neighbors of its neighbors (e.g. \cite{bianconi2014pre}).  
\\[15pt]
\end{itemize}

\noindent {\bf Acknowledgments and Financial Support}\\

\noindent Authors acknowledge support from CNR PNR Project
``CRISIS Lab''. \\

%\newpage

\appendix

\section{Appendix}
\label{appendix}

\subsection{Proof of the asymptotic 
behavior of $L_n$}

\begin{theorem}\label{th-L}
Consider our model, the
following statements hold true:
\begin{itemize}
\item[a)] ${L_n}/{\ln(n)}\stackrel{a.s.}\longrightarrow 
\alpha$ for $\beta =0$;
\item[b)] ${L_n}/{n^{\beta}}\stackrel{a.s.}\longrightarrow
{\alpha}/\beta$ for $\beta\in (0,1]$.
\end{itemize}
\end{theorem}

\begin{proof}  
Set
$\lambda_1=\alpha$ and recall that the random variables 
$N_i$ are independent and each $N_{i}$ has 
distribution Poi$(\lambda_{i})$.

\indent The assertion b) is trivial for $\beta=1$ since, 
in
this case, $L_n$ is the sum of $n$ independent random 
variables with
distribution $\hbox{Poi}(\alpha)$ and so, 
by the classical strong
law of large numbers, 
$L_n/n\stackrel{a.s.}\longrightarrow\alpha$. 

\indent Now, let us prove assertions a) and b) 
for $\beta\in
        [0,1)$. Define
\begin{gather*}
\lambda(\beta)=\alpha\,\text{ if }\,\beta=0\quad\text{and}
\quad\lambda(\beta)=\frac{\alpha}{\beta}\,\,\text{ if }\,
\beta\in (0,1),
\\
a_n(\beta)=\log{n}\,\text{ if }\,\beta=0\quad\text{and}
\quad 
a_n(\beta)=n^\beta\,\text{ if }\,\beta\in (0,1).
\end{gather*}
\noindent We need to prove that 
$L_n/a_n(\beta)\stackrel{a.s.}\longrightarrow
\lambda(\beta)$. 
First, we observe that 
\begin{gather*}
\frac{\sum_{i=1}^{n}\lambda_i}{a_n(\beta)}
=
\alpha\frac{\sum_{i=1}^{n} i^{\beta-1}}
{a_n(\beta)}
\longrightarrow \lambda(\beta),
\end{gather*}
\noindent Next, let us define
\begin{gather*}
T_0=0\quad\text{and}\quad 
T_n=\sum_{i=1}^n\frac{N_i-E[N_i]}{a_i(\beta)}=
\sum_{i=1}^n\frac{N_i-\lambda_{i}}{a_i(\beta)}.
\end{gather*}
\noindent Then $(T_n)$ is a martingale with 
\begin{gather*}
E[T_n^2]=
\sum_{i=1}^n
\frac{E\bigl[(N_i-\lambda_{i})^2\bigr]}{a_i(\beta)^2}=
\sum_{i=1}^n\frac{\lambda_{i}}{a_i(\beta)^2}
\end{gather*}
and so 
$\sup_n E[T_n^2]=\sum_{i=1}^{+\infty}\frac{\lambda_{i}}{a_i(\beta)^2}<+\infty$. 
Thus, $(T_n)$ converges a.s. and the Kronecker's lemma 
implies
$$
\frac{1}{a_n(\beta)}\sum_{i=1}^n a_i(\beta)
\frac{(N_i-\lambda_{i})}{a_i(\beta)}
\stackrel{a.s.}\longrightarrow 0,
$$
that is
$$
\frac{\sum_{i=1}^{n}N_i}{a_n(\beta)}-\frac{\sum_{i=1}^n\lambda_i}{a_n(\beta)}
\stackrel{a.s.}\longrightarrow 0.
$$
\noindent Therefore, we can conclude that 
\begin{gather*}
\lim_n\frac{L_n}{a_n(\beta)}=
\lim_n\frac{\sum_{i=1}^n N_i}{a_n(\beta)}=
\lim_n\frac{\sum_{i=1}^n\lambda_{i}}{a_n(\beta)}=
\lambda(\beta)\quad\text{a.s.}
\end{gather*}

\end{proof}

\begin{remark} \rm The above Theorem implies 
that $\ln(L_n)/\ln(n)$ is a strongly consistent 
estimator of $\beta$. Indeed, if $\beta=0$ 
then $L_n\stackrel{a.s.}\sim \alpha\ln(n)$ as $n\to
    +\infty$; hence 
$\ln(L_n)\stackrel{a.s.}\sim\ln(\alpha) +
    \ln(\ln(n))$, therefore
    $\ln(L_n)/\ln(n)\stackrel{a.s.}\sim\ln(\alpha)/\ln(n) +
    \ln(\ln(n))/\ln(n)\stackrel{a.s.}\to 0=\beta$. 
Furthermore, if $\beta>0$, then we have 
$L_n\stackrel{a.s.}\sim (\alpha/\beta)
    n^{\beta}$ as $n\to +\infty$ so
    $\ln(L_n)\stackrel{a.s.}\sim\ln(\alpha/\beta) + 
\beta \ln(n)$,
   hence
    $\ln(L_n)/\ln(n)\stackrel{a.s.}
\sim\ln(\alpha/\beta)/\ln(n) +
    \beta\stackrel{a.s.}\to \beta$.
\end{remark}

\subsection{Simulations of the unipartite network: some analysis 
on its structure}

We generated feature matrices with $n=1\,000$ nodes taking fixed
values for $\alpha$ and $\beta$, i.e. $\alpha=10$ and $\beta=0.5$, and
different values for $\delta$ ($\delta\in [0.1,\,0.5]$).  Starting
from these feature matrices, we considered the structure of the
unipartite network for three different values of $K$ ($K=1,\,4,\,10$)
and three different values of $p$ ($p=0,\,0.1,\,0.5$). \\

\indent We considered the following quantities:
\begin{itemize}
\item the clustering coefficient defined as:
\begin{align}\label{CC}
C=\frac
{\hbox{Number of closed triplets}}
{\hbox{Number of connected triplets of nodes}},
\end{align} 
 where a connected triplet is a set of three nodes that are connected
 by two or three undirected links (open and closed triplet,
 respectively). See Table \ref{tab_tri}.

\item the fraction of pairs of nodes at distance at most $20$,
  i.e. the fraction of pairs of nodes that are reachable from each other
  within at most $20$ steps (see Table
  \ref{tab_distances}):
\begin{equation}\label{RP_20}
RP_{20}=\frac
{\hbox{Number of couples of nodes at distance at most $20$}}
{\hbox{Number of couples of nodes}}.
\end{equation}
We recorded also the observed maximum value $h^*$ of the distance
between the nodes.

\item the degree distribution, in the sense of the Complementary
  Cumulative Distribution Function (CCDF) of the number of neighbors
  of each node (see Figure \ref{degree}).
\end{itemize}

\indent The clustering coefficient $C$ strongly increases with $p$ (as
expected). For $p=0$ the percentage of closed triplets increases with
$\delta$, but remains smaller or equal than $13\%$ of total triplets
for all considered values of $\delta$ and $K$. For values of $p$
greater than zero, the percentage of closed triplets increases with
$\delta$ in a range of $13\%-30\%$ for $p=0.1$ and in a range of
$39\%-62\%$ for $p=0.5$.  The effect of $K$ and $\delta$ seems to be
marginal on the clustering coefficient.  \\

\begin{table}[ht]
 \centering
{\scriptsize{
 \begin{tabular}{l | l | c c c c c c c c c c }

& &$\delta\, =$ & 0.1 & 0.15 & 0.2 & 0.25 & 0.3 & 0.35 & 0.4 & 0.45 & 0.5 \\ 
  \hline
& $p\,=\,0$&   & 0.04& 0.05 &  0.07 &  0.08 & 0.08 & 0.10 
&  0.13&  0.13& 0.10 \\
 $K=1$ 
& $p\,=\,0.1$&  & 0.13& 0.17& 0.20 & 0.23&  0.23&  0.24
& 0.26& 0.27& 0.30 \\
&  $p\,=\,0.5$&& 0.39& 0.45& 0.45& 0.49 & 0.49& 0.47&   0.49
&  0.53& 0.62 \\ 
   \hline
&  $p\,=\,0$&   & 0.06& 0.06 & 0.08 & 0.09 & 0.08 & 0.11 
& 0.13& 0.13& 0.11 \\
$K=4$
& $p\,=\,0.1$&  &  0.15& 0.18& 0.21 &  0.24& 0.23& 0.25& 0.26
& 0.28& 0.30 \\
&   $p\,=\,0.5$& & 0.42& 0.47& 0.46& 0.49 & 0.49& 0.48
&  0.50&  0.53& 0.62 \\
  \hline 
&  $p\,=\,0$&   & 0.06&  0.06 &  0.08 & 0.09 & 0.08 
& 0.11 &  0.13& 0.14& 0.11 \\
$K=10$ 
&  $p\,=\,0.1$&  & 0.15& 0.18& 0.21 & 0.24& 0.23
&  0.25& 0.26& 0.28& 0.30 \\
&   $p\,=\,0.5$& & 0.42& 0.47&  0.46& 0.49 &  0.49
& 0.48&  0.49 &  0.53& 0.62 \\
   \hline
 \end{tabular}
  }}
 \caption{Clustering coefficient (averaged over $100$
   realizations) for $\alpha=10$, $\beta=0.5$, $\ell=4000$, and
   different values of $\delta$, $K$, and $p$.}
\label{tab_tri} 
\end{table}

\indent Looking at the values obtained for the fraction of pairs of
nodes at distance at most $20$, for the two different values $\delta=0.1$ 
and $\delta=0.5$, we can notice a clear difference in the behavior 
(independently of $K$ and $p$): indeed, the fraction of reachable 
pairs for $\delta=0.1$ (when $K$ and $p$ are fixed) is highly greater
than the corresponding fraction for $\delta=0.5$.  Moreover, the fraction
of reachable pairs decreases when $K$ increases (and the other parameters
are fixed) and slightly changes when only $p$ varies. 
The complementary fraction corresponds to the pairs of nodes at distance 
greater than $20$ or not reachable from each other.
\\ \indent The observed maximum distance $h^*$ (among pairs
of nodes at distance at most $20$) varies in range of $2-5$ and
decreases when $\delta$ ($p$ and $K$, respectively) increases and the
other parameters are fixed.\\

\begin{table}[ht]
 \centering
{\scriptsize{
 \begin{tabular}{l | c c c| c c| c c }
 \hline
  & &$K=1$ & & $K=4$ & & $K=10$&\\
\hline
   &$\delta=$ & 0.1 & 0.5 & 0.1 & 0.5 & 0.1 & 0.5  \\ 
 \hline
$p=0$&   &  0.439 ($5$) & 0.128 ($4$) & 0.350 ($4$) 
& 0.118 ($4$) & 0.349 ($4$)& 0.117 ($4$) \\

$p=0.1$&  & 0.438 ($4$)&  0.128 ($3$) 
& 0.352 ($3$)& 0.118 ($3$)& 0.350 ($3$) & 0.117 ($3$)\\
  
$p=0.5$&  & 0.437 ($3$)& 0.128 ($2$)& 0.351 ($2$) & 
0.118 ($2$)&  0.349 ($2$) & 0.117 ($2$)\\
 \end{tabular}
  }}
\caption{Fraction of pairs of nodes at distance at most $20$ 
  (averaged over $100$ realizations) for $\alpha=10$, $\beta=0.5$,
  $\ell=4000$, and different values of $\delta$, $K$, and $p$. For each
  set of parameters, the corresponding observed maximum distance $h^*$
  is reported in brackets.}
\label{tab_distances} 
\end{table}

\indent Finally, the effect of $p$ on the total number of links is
clear: when $p=0$ the number of links is approximately equal to the
chosen $\ell$ (i.e. $\ell=4000$), since in this case we have only the
first phase of the unipartite network construction: links are related
only to the features. The larger $p$ the more triangles are closed and
so the more links we have. Table \ref{tab_links} reports the total
number of links for all combinations of the parameters.  Regarding the
degree distribution, Figure \ref{degree} shows the CCDF of the number
of neighbors of a node.  Parameter $p$ also
influences the shape of the degree distribution, together with
$\delta$ and $K$.\\

\begin{table}[ht]
 \centering             
{\scriptsize{
 \begin{tabular}{l | c c c| c c| c c  }
 \hline
  & &$K=1$ & & $K=4$ & & $K=10$&\\
\hline    
  &$\delta\, =$ & 0.1 & 0.5 & 0.1 & 0.5 & 0.1 & 0.5  \\ 
  \hline
 $p\,=\,0$&   &   4\,003.47 & 3\,998.15  &  4\,002.17 
&  3\,999.59 & 3\,997.13 &  3\,999.52 \\

 $p\,=\,0.1$& & 17\,853.46& 19\,862.54 & 19\,107.53
& 19\,523.42& 19\,112.46 & 19\,484.86 \\

 $p\,=\,0.5$&&  93\,093.05& 43\,538.68& 81\,343.97& 
41\,382.62&  81\,039.49  &  41\,156.34\\
 \end{tabular}
  }}
\caption{Total number of links in the unipartite network 
  (averaged over $100$ realizations) for $\alpha=10$, $\beta=0.5$,
  $\ell=4000$, and $\delta$, $K$, and $p$ varying. Note that for $p=0$
  the number is around the chosen $\ell=4000$.}
\label{tab_links} 
\end{table}

\begin{figure}
\centering
\begin{subfigure}[b]{0.30\textwidth}
\includegraphics[width=\textwidth]{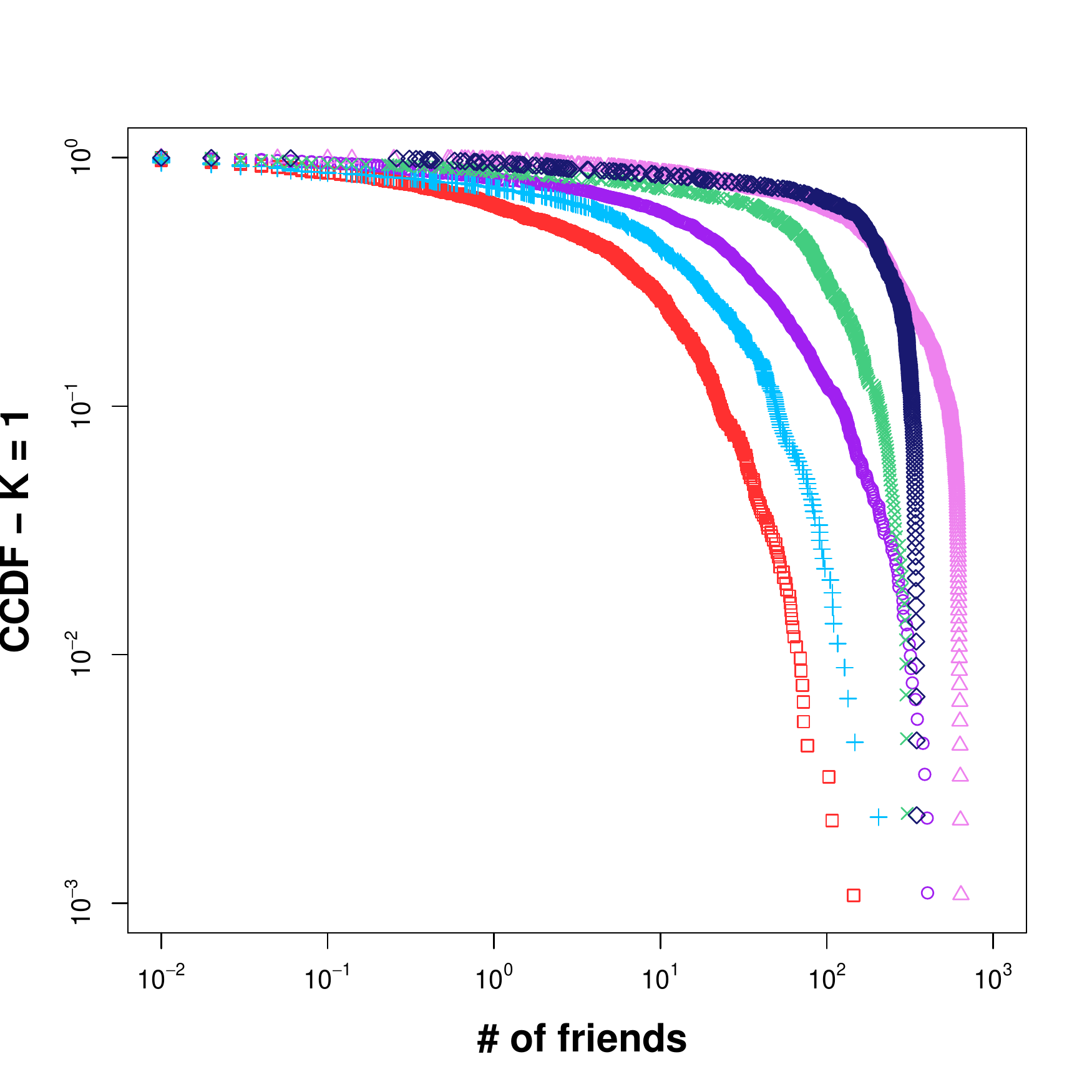}
\end{subfigure}
\begin{subfigure}[b]{0.30\textwidth}
\includegraphics[width=\textwidth]{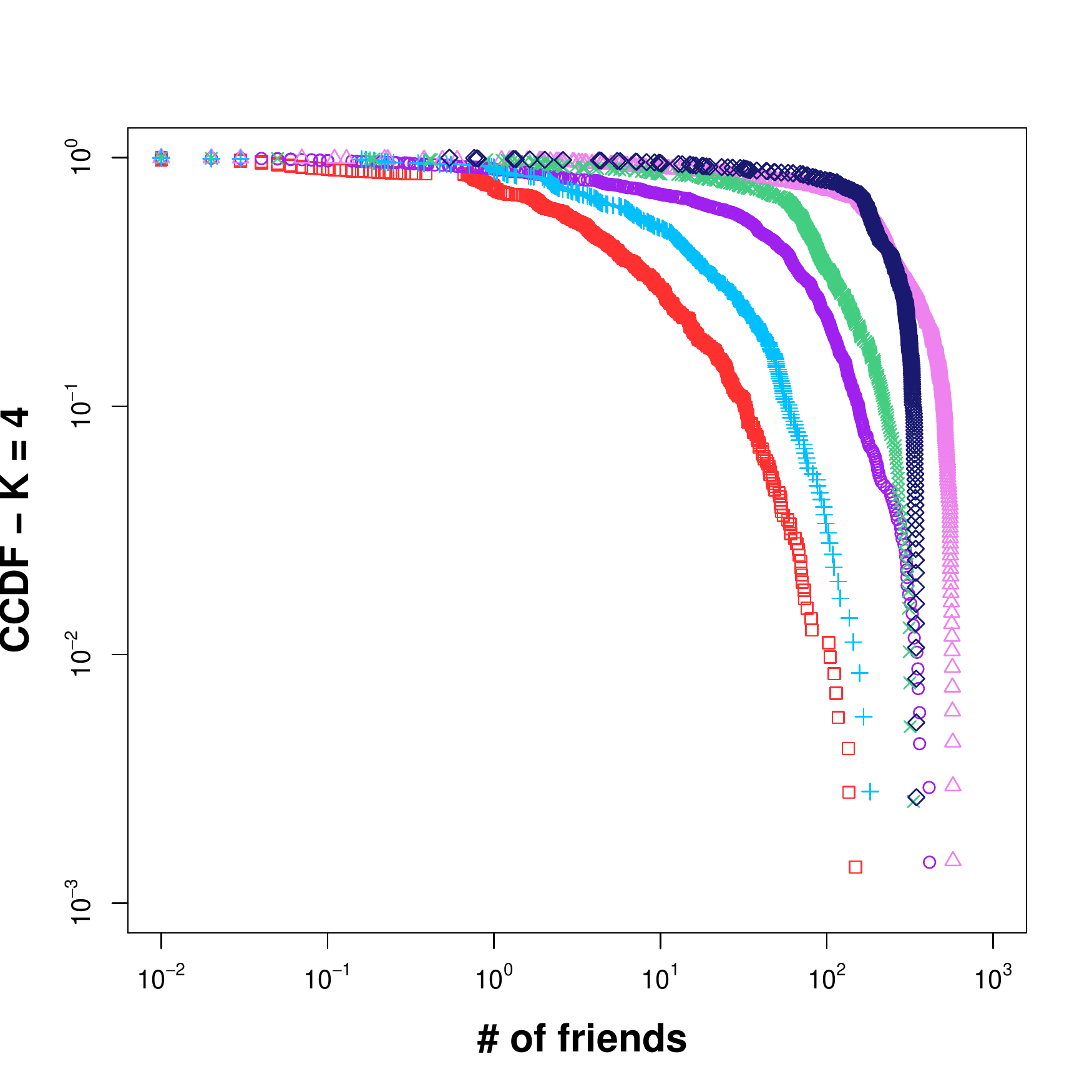}
\end{subfigure}
\begin{subfigure}[b]{0.30\textwidth}
\includegraphics[width=\textwidth]{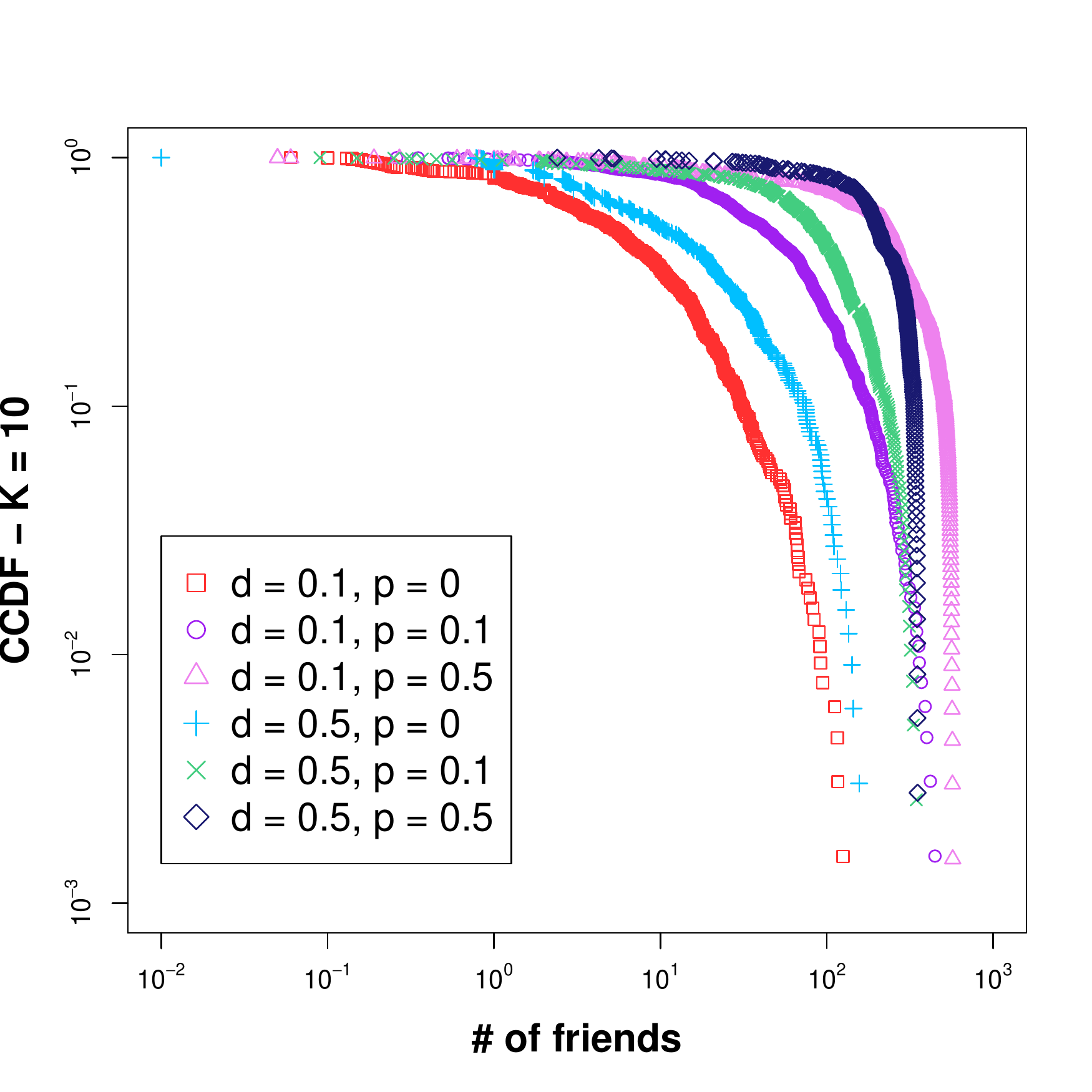}
\end{subfigure}
\caption{CCDF of the number of neighbors (averaged over $100$
  realizations) for $\alpha=10$, $\beta=0.5$, $\ell=4000$, and
  different values of $K$ (corresponding to different boxes) and
  different values of $\delta$ and $p$ (corresponding to different
  symbols and colors).  }
\label{degree}
\end{figure}

\end{document}